\documentclass[12pt]{article}
\input xy
\xyoption{all}
\usepackage{a4wide}
\usepackage[top=2cm,bottom=3cm,left=2cm,right=2cm]{geometry}
\usepackage[latin1]{inputenc}
\usepackage{amsfonts}
\usepackage{amsmath}
\usepackage{amssymb}
\usepackage{amsthm}
\usepackage{verbatim} 
\usepackage{mathtools}
\usepackage{mathrsfs}  
\usepackage{euscript}  
\usepackage[scr=rsfso]{mathalfa}
\usepackage{graphicx}
\usepackage{hyperref}
\hypersetup{colorlinks=true,urlcolor=blue,linkcolor=black,citecolor=black}

\usepackage{float} 
\usepackage{caption}
\usepackage{chngcntr}
\hypersetup{colorlinks=true,urlcolor=blue,linkcolor=black,citecolor=black}

\usepackage[english]{babel}
\usepackage{color}

\newtheorem{Theorem}{Theorem}[section]  
\newtheorem{Lemma}[Theorem]{Lemma}

\newtheorem{Proposition}[Theorem]{Proposition}
\newtheorem{Definition}[Theorem]{Definition}
\newtheorem{Remark}[Theorem]{Remark}

\numberwithin{equation}{section}
\counterwithout{figure}{section} 

\def\cc#1{\mathcal{#1}}

\newcommand{\N}{{\mathbb N}}

\newcommand{\R}{{\mathbb R}}
\newcommand{\DD}{{\mathbb D}}   
\newcommand{\BB}{{\mathbb B}}   
   
\newcommand{\WW}{{\mathbb W}}   
\newcommand{\Z}{{\mathbb Z}}
\newcommand{\Co}{{\mathbb C}}

\newcommand{\sX}{{\mathsf X}}   
\newcommand{\sY}{{\mathsf Y}}   

\newcommand{\kD}{\mathfrak D}   
\newcommand{\kW}{\mathfrak W}   
\newcommand{\kG}{\mathfrak S}   

\newcommand{\cK}{\mathcal K}
\newcommand{\cL}{\mathcal L}
\newcommand{\cS}{{\mathcal S}}

\newcommand{\cP}{{\mathcal P}}
  
\newcommand{\cV}{\mathcal V}

\newcommand{\crF}{{\mathscr F}}  
\newcommand{\crD}{{\mathscr D}}  
\newcommand{\crW}{{\mathscr W}}  

\newcommand{\EuM}{\EuScript M}

\newcommand{\tti}{{\mathtt i}}
\newcommand{\tto}{{\mathtt o}}
\newcommand{\ttb}{{\mathtt b}}
\newcommand{\fF}{D}

\newcommand{\I}[2]{I^{\tiny \text{\it #1-ary}}_{#2}}

\newcommand{\uB}{{\underline B}}
\newcommand{\ux}{{\underline x}}

\newcommand{\uz}{\underline z}

\newcommand{\fh}{\phi}
\newcommand{\Fh}{\Phi}

\newcommand{\nn}{\nonumber}
\newcommand{\one}{{\mathbf 1}}

\newcommand{\tW}{\WW}
\newcommand{\Wro}{\WW^{\rot}}

\newcommand{\tI}{\tilde I}
\newcommand{\tO}{\tilde O}
\newcommand{\tg}{\tilde g}
\newcommand{\thh}{\tilde h}
\newcommand{\bh}{\bar h}
\newcommand{\bg}{\bar g}

\newcommand{\tw}{\tilde w}

\newcommand{\muq}{Q}

\newcommand{\tsum}{{\textstyle{\sum}}}
\newcommand{\tprod}{{\textstyle{\prod}}}
\newcommand{\tint}{{\textstyle{\int}}}

\DeclareMathOperator\perm{perm}
\DeclareMathOperator\Det{Det}
\DeclareMathOperator{\free}{free}
\DeclareMathOperator{\rot}{root}

\DeclareMathOperator{\ls}{ls}
\DeclareMathOperator{\capa}{cap}
\DeclareMathOperator{\unif}{unif}
\DeclareMathOperator{\boson}{boson}

\DeclareMathOperator{\ri}{ri}
\DeclareMathOperator{\Tr}{Tr}
\DeclareMathOperator{\st}{ST}

\DeclareMathOperator{\TI}{TI}

\def\quotient#1#2{%
    \raise1ex\hbox{$#1$}/\lower1ex\hbox{$#2$}%
}
\allowdisplaybreaks

\begin{document}

\title{Gaussian random permutation\\ and the boson point process}
\date{}
\author{In\'es Armend\'ariz\thanks{Universidad de Buenos Aires \& IMAS-UBA-CONICET, iarmend@dm.uba.ar},\quad
Pablo A. Ferrari\thanks{Universidad de Buenos Aires \& IMAS-UBA-CONICET, pferrari@dm.uba.ar},\quad
Sergio Yuhjtman\thanks{IMAS-UBA-CONICET, syuhjtma@dm.uba.ar}
}
\maketitle

\begin{abstract}
  We construct an infinite volume spatial random permutation $(\sX,\sigma)$, where $\sX\subset\mathbb R^d$ is locally finite and $\sigma:\sX\to \sX$ is a permutation, associated to the formal Hamiltonian 
  \begin{align}
    \nonumber
    H(\sX,\sigma) = \sum_{x\in \sX} \|x-\sigma(x)\|^2.
  \end{align}
  The measures are parametrized by the point density $\rho$ and the temperature $\alpha$.    
Spatial random permutations are naturally related to boson systems through a representation originally due to Feynman (1953). Let $\rho_c=\rho_c(\alpha)$ be the critical density for Bose-Einstein condensation in Feynman's representation.
Each finite cycle of $\sigma$ induces a loop of points of~$\sX$.
 For $\rho\le \rho_c$ we define $(\sX, \sigma)$ as a Poisson process of finite unrooted loops of a random walk with Gaussian increments that we call Gaussian loop soup, analogous to the Brownian loop soup of Lawler and Werner (2004). We also construct Gaussian random interlacements, a Poisson process of doubly infinite trajectories of random walks with Gaussian increments analogous to the Brownian random interlacements of Sznitman (2010).
  For $d\ge 3$ and $\rho>\rho_c$ we define $(\sX,\sigma)$ as the superposition of independent realizations of the Gaussian loop soup at density $\rho_c$ and the Gaussian random interlacements at density $\rho-\rho_c$. In either case we call  $(\sX, \sigma)$ a Gaussian random permutation at density $\rho$ and temperature $\alpha$. The resulting measure satisfies a Markov property and it is Gibbs for the Hamiltonian $H$.  Its point marginal  $\sX$ has the same distribution as the boson point process introduced by Shirai-Takahashi (2003) in the subcritical case, and by Tamura-Ito (2007) in the supercritical case. 
\smallskip
\smallskip


\smallskip

\smallskip


  \end{abstract}

\section{Introduction} 

The free Bose gas has been intensively studied from different perspectives
in mathematical physics. Araki-Woods \cite{AW} and Cannon \cite{Cannon} construct it as an infinite volume quantum model; in these works the free Bose gas 
at density $\rho$ is given by a state $\varphi_\rho$
of an appropriate $C^*$-algebra $\cc A$ of observables.
In a finite volume box $\Lambda \subset \R^d$, the free Bose gas of $N$ particles at temperature $\alpha$ is described as a quantum system on 
the Hilbert space of symmetric functions
$L^2_s(\Lambda^N)$. The Hamiltonian is minus the Laplacian multiplied by a positive constant
that we may choose as $1$.
The behavior of the grand-canonical ensemble 
as $\Lambda \to \R^d$, $N/|\Lambda| = \rho$, where $|\Lambda|$ is the volume of $\Lambda$, was studied
by Einstein by diagonalizing the Hamiltonian, reaching the famous
conclusion that when $d\ge 3$ there is a critical density above which 
the lowest energy level exhibits macroscopic occupation number.
This critical density is
\begin{equation}
\label{zeta}
\rho_c(\alpha) = 
\Bigl(\frac{\alpha}{\pi}\Bigr)^\frac{d}{2} \sum_{k\ge 1}\frac{1}{k^{\frac{d}{2}}}.
\end{equation}
Instead of diagonalizing the Hamiltonian, we may write
the partition function as
\begin{align}
 \;Z_{\Lambda,N}\;&:=\;\Tr_{L^2_s(\Lambda^N)}(e^{-\frac{1}{2\alpha} H}) \;=\; 
\Tr_{L^2(\Lambda^N)}(P_+ e^{-\frac{1}{2\alpha} \Delta})\\[2mm]
&=\frac{(\alpha/\pi)^{Nd/2}}{N!} \sum_{\sigma \in \cS_N} 
\int_{\Lambda^N} \ e^{-\alpha \sum_i \|x_i-x_{\sigma(i)}\|^2}\, d\ux, \label{Zlambda}
\end{align}
where $\Delta$ is the Laplacian, $P_+$ is the symmetrization operator, $\ux=(x_1,\dots,x_N)$  and $\cS_N$ denotes the set of permutations of $\{1,\dots,N\}$. We abuse notation by writing $\sigma(x_i) = x_{\sigma(i)}$. 
The formula can be obtained by integrating the kernel of the integral operator
$P_+ e^{-\frac{1}{2\alpha} \Delta}$ along the diagonal. This partition function coincides with a statistical mechanics model of spatial random permutations on the configuration space 
$\Lambda^N \times \cS_N$ and Radon-Nikodym density 
\begin{equation}\label{finite-box}
f (\ux,\sigma) = \frac1{Z_{\Lambda,N}}\, \frac{(\alpha/\pi)^{Nd/2}}{N!}\, e^{-\alpha \sum_{i=1}^N \|  x_i-\sigma(x_i)  \|^2}
\end{equation}
with respect to the Lebesgue measure on $\Lambda^N$ and the uniform counting
measure on~$\cS_N$. 
This representation was first proposed by Feynman \cite{Feynman} in 1953  in order to explain the transition 
of Helium-4 from fluid to superfluid.
 Feynman claimed that the transition to the superfluid phase 
coincides with the appearance of infinite cycles in a typical spatial random permutation.
The approach was taken up by S\"ut\H{o} \cite{Suto1}, \cite{Suto2},  who proved Feynman's claim for Bose-Einstein condensation of the free Bose gas.

The thermodynamic limit of the point marginal of the spatial random permutation associated to \eqref{finite-box} with $N\approx\rho |\Lambda|$ has been considered by several authors. In the subcritical case when the point density is $\rho\le \rho_c$, Tamura and Ito \cite{TI}  
identified the limit with the \emph{boson point process}
studied by Macchi \cite{Macchi} and Shirai and Takahashi \cite{ST}, a process with $n$-point correlation functions
\begin{align} \label{loop-correlations}
  \varphi_n(x_1,\dots,x_n) &= \perm \bigl(K_\lambda(x_i,x_j)\bigr)_{i,j=1}^n,\\[2mm]
  K_\lambda(x,y)&:= \sum_{k \geq 1} \Bigl(\frac{\alpha}{\pi k}\Bigr)^{d/2}\,\lambda^k\,e^{-\frac{\alpha}{k} \|x-y\|^2}, \label{Klambda}
\end{align}
where the parameter
$\lambda \in (0,1)$, if $d\le 2$, and $\lambda \in (0,1]$, if $d\ge 3$, is an increasing function of the density $\rho$, and $\perm(K)$ denotes the permanent of the matrix $K\in \R^{n\times n}$. Since the series \eqref{Klambda} diverges for $\lambda>1$, the process with these correlations is well defined for densities $\rho<\infty$ in dimensions $1$ and $2$, and $\rho\le\rho_c$, $\rho_c$ the solution of $\lambda (\rho_c) =1$, if $d\ge 3$.
Tamura-Ito  \cite{TI2} consider the supercritical boson point process at density $\rho>\rho_c$ in dimension $d\ge 3$ and show that it consists of the superposition of the critical boson point process at density $\rho_c$ with another point process at density $\rho-\rho_c$, see \S\ref{pm2} for a precise description of the latter.

In this paper, to every positive density $\rho$ and temperature $\alpha$ we associate an infinite volume spatial random permutation $(\sX,\sigma)$, where $\sX$ is a discrete subset of $\R^d$ and $\sigma:\sX \to \sX$ is a permutation, that is, a bijection. The law of $(\sX,\sigma)$ is translation-invariant and has point density $\rho$.
Our construction stems from the observation that the density \eqref{finite-box} can be written as a product of weights assigned to the loops $\gamma$ induced by the cycles of the permutation~$\sigma$. An unrooted \emph{loop} of size $k$ is described as $\gamma=[x_1,\dots,x_k]$ with $x_i\in\R^d$; the square brackets indicate that $[x_2,\dots,x_k,x_1]=[x_1,\dots,x_k]$. We denote by $\{\gamma\}=\{x_1,\dots,x_k\}$ the set of points in the loop, and write $\gamma(x_i)=x_{i+1}$. A spatial permutation $\Gamma= (\sX, \sigma)$, when $\sX$ is finite, can be decomposed in loops induced by the cycles of $\sigma$.
 We will say that $\gamma\in\Gamma$ if $\gamma$ is a loop with $\{\gamma\}\subset \sX$ and $\gamma(x)=\sigma(x)$ for all  $x\in\{\gamma\}$; with this notation we have $\bigcup_{\gamma\in\Gamma}\{\gamma\}=\sX$. The numerator of the density function \eqref{finite-box} factorizes as follows
\begin{align}
  \label{independence}
  e^{-\alpha \sum_{i=1}^N \| \sigma(x_i) - x_i \|^2} = \prod_{\gamma\in\Gamma} e^{-\alpha \sum_{x\in\{\gamma\}} \| \gamma(x) - x \|^2}.
\end{align}
The independence of loops suggested by \eqref{independence} was already present in S\"ut\H{o} \cite{Suto1}.
\\

     \noindent\begin{minipage}{\linewidth}\begin{center}
    \includegraphics[scale=.4]{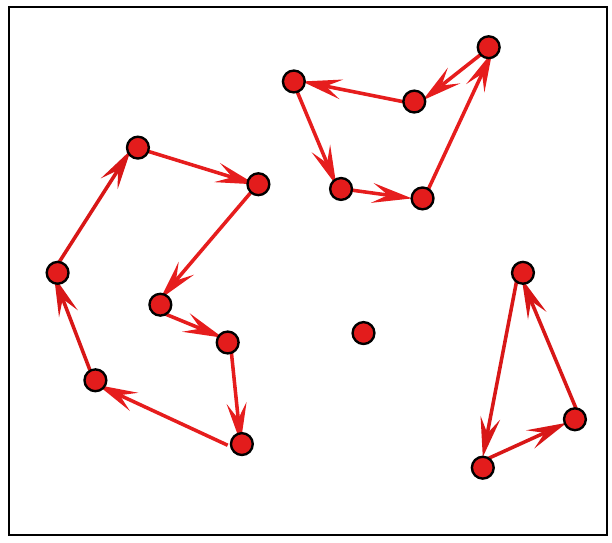}%
\\
\captionof{figure}{Loops induced by a spatial random permutation in a box. An arrow from $x$ to $y$ means $y=\sigma(x)$. Isolated dots correspond to points $x=\sigma(x)$, loops of length 1. \label{gibbs-2}}
\end{center}\end{minipage}

We define the \emph{Gaussian loop soup} as a Poisson process of unrooted loops with Gaussian increments analogous to the Brownian loop soup introduced by Lawler and Werner in \cite{LW}, see also Lawler and Trujillo Ferreras
 \cite{L-TF} and Le Jan \cite{LeJan}. The process was previously considered by Fichtner \cite{MR1121197,Wakolbinger}, who related it to the ideal Bose Gas. We consider the loop soup at density $\rho\le \rho_c$.
When $d\ge 3$ and the point density $\rho$ exceeds the critical density $\rho_c$, we allocate the leftover density $\rho-\rho_c$ to an independent realization of \emph{Gaussian random interlacements}, a Poisson process of doubly infinite trajectories of a discrete-time random walk with Gaussian increments. This process is the discrete-time version of the Brownian random interlacements introduced by Sznitman \cite{Sznitman1}.
\\
     \noindent\begin{minipage}{\linewidth}\begin{center}
    \includegraphics[scale=.4]{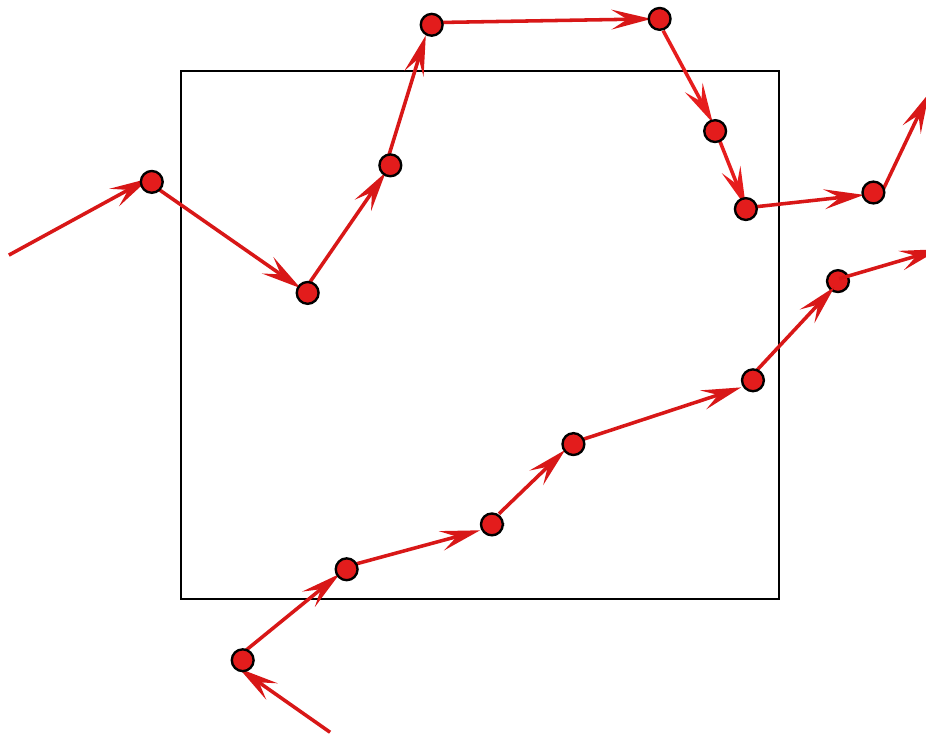}%
\\
\captionof{figure}{Random interlacements intersecting a box. \label{gibbs-23}}
\end{center}\end{minipage}

We define the {\it Gaussian random permutation in $\R^d$ with temperature $\alpha$ and density $\rho$} as the spatial random permutation obtained by superposing independent realizations of a Gaussian loop soup with increments of variance $\frac{1}{2\alpha}$ 
at point density $\rho\wedge \rho_c$, and a Gaussian random interlacements at point density $(\rho-\rho_c)^+$ with the same increments.
\\
    \noindent\begin{minipage}{\linewidth}\begin{center}
    \includegraphics[scale=.4]{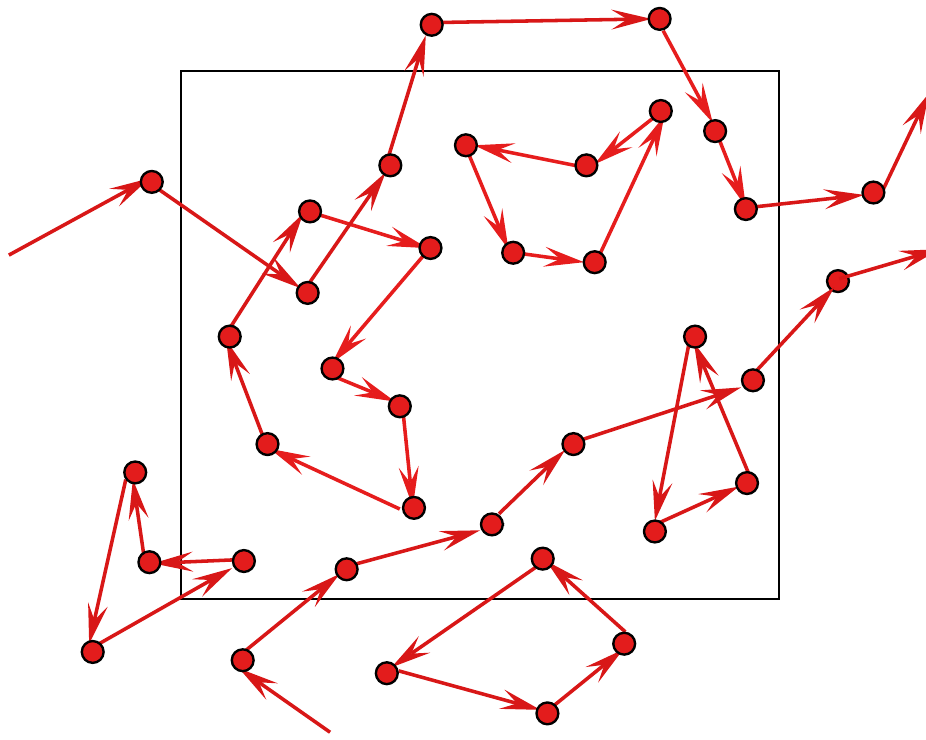}%
\\
\captionof{figure}{Loop soup plus random interlacements intersecting a box at supercritical density $\rho>\rho_c$. \label{gibbs-22} }
\end{center}
\end{minipage}

Our main result can now be stated as follows.
\begin{Theorem}
  \label{t1}
For  any dimension $d$, temperature $\alpha>0$ and point density $\rho>0$, the Gaussian random permutation with these parameters satisfies a Markov property, and is a Gibbs measure for the specification induced by the finite-volume distribution \eqref{finite-box}. Furthermore, the point marginal of this measure is the boson point process.
\end{Theorem}

This result follows from  propositions \ref{markov}, \ref{t43} and \ref{teo1}. In this context, the Markov property means that when (a) the points in $\Lambda^c$ and the arrows connecting points in $\Lambda^c$, and (b) the arrows pointing into or out of $\Lambda$, are given as a boundary condition, the distribution of the spatial permutation inside $\Lambda$ only depends on the positions of those points in $\Lambda^c$ such that the arrows emanating from them point into $\Lambda$, or the arrows pointing to them originate in $\Lambda$, and the directions of these arrows; see Fig.~\ref{gibbs-4} in \S\ref{grp}. 

We show in Proposition \ref{ls-correl} that the point marginal of the loop soup is a permanental process, with correlations \eqref{loop-correlations}. 
This result also follows from the equivalence between the point marginal of the loop soup and the subcritical boson point process shown in \S\ref{s4} and the results in \cite{TI}. It can be understood as a simple version of the identification of the occupation field moments of loop soups as permanents, see \cite{MR2815763}.

The boson point process is a Cox process, that is, a Poisson process with a random intensity measure. In the subcritical case, the result is a consequence of the fact that the boson point process is permanental, see \eqref{loop-correlations}. The property was established for the supercritical case by Eisenbaum in \cite{Eis}. With the representation of the boson point process as the point marginal of the loop soup or the superposition of the loop soup and the random interlacements, the property can be interpreted as an instance of the isomorphism theorems derived by Le Jan~\cite{MR2675000} and Sznitman~\cite{MR2892408}.

There are several results concerning the marginal distribution of the permutations. Benfatto, Cassandro, Merola and Presutti~\cite{MR2125575} proved that in the thermodynamic limit of the supercritical case, the maximal cycle-size rescales with the volume. Betz and Ueltschi~\cite{BU-PD} prove a phase transition for a family of annealed models that includes the one described by \eqref{finite-box}. 
 Precisely, they prove that above a critical density that equals $\rho_c$ for the model considered here, a typical permutation in a box of volume $N/\rho$ contains macroscopic cycles, and the sequence of their rescaled lengths, when sorted in decreasing order, converges to a Poisson-Dirichlet distribution. These results were recently extended to a different class of models by  Bogachev and Zeindler \cite{zbMATH06410137} and  Elboim and Peled~\cite{Elboim2019}.

Our results refer to the annealed model, where configurations of points 
 and permutations of these points are jointly sampled. Spatial permutations over fixed sets of points have also been studied. 
An early result due to Fichtner~\cite{Fichtner} in 1991, and more recently works by Gandolfo, Ruiz and Ueltschi~\cite{GRU}, Betz and Ueltschi~\cite{UB}, Betz~\cite{B}, Biskup and Richthammer~\cite{B-R} and Armend\'ariz, Ferrari, Groisman and Leonardi~\cite{AFGL}, consider these issues when the set of points is fixed as either the regular lattice or a realization of a locally finite point process. It is expected that for dimensions $d\ge 3$ there is a critical temperature below which, almost surely on the realization of the point process, a typical infinite volume permutation will contain infinite cycles. This is a challenging problem, at present all rigorous results pertain to the high temperature regime.  

 We expect that the Gaussian random permutation with temperature $\alpha$ and density $\rho$ is the thermodynamic limit of the distribution in \eqref{finite-box}, and that infinite cycles evolve out of macroscopic cycles. In particular, this would provide an alternative proof that the Gaussian random permutation is a Gibbs measure for the specification associated to \eqref{finite-box}. Vogel~\cite{vogel} has recently derived the thermodynamic limit in a discrete space setting.

The rest of the article is structured as follows. We construct the Gaussian loop soup in \S\ref{GLS}. In \S\ref{fls5} we identify the Gaussian loop soup in finite volume conditioned to having a fixed number of points as the canonical measure of the spatial random permutation associated to \eqref{finite-box}, and the unconditioned Gaussian loop soup in finite volume as the grand-canonical measure. We then compute the Laplace functionals of the Gaussian loop soup point marginal and its $n$-point correlations in \S\ref{pm1}. In \S\ref{sri5} we introduce Gaussian random interlacements, and derive the Laplace functionals and $n$-point correlations of the point marginal in \S\ref{pm2}. These Laplace functionals and $n$-point correlations are employed in \S\ref{s4} to prove that the boson point process, in any dimension and at any density, is in fact the point marginal of the Gaussian Random permutation with the same parameters; the latter is formally defined in \S\ref{infvol}. We finally show in \S\ref{grp} that the Gaussian random permutation has the Markov property, and moreover it is a Gibbs measure for the specification determined by \eqref{finite-box}.

\section{Gaussian loop soup}\label{GLS}

In this section we construct the Gaussian loop soup. 
We first introduce the Gaussian loop soup following the construction of the simple random walk loop soup by Lawler and Trujillo Ferreras~\cite{L-TF}. In \S\ref{fls5} we show that the loop soup restricted to a bounded set $\Lambda$ and conditioned to having $N$ points is equal to the canonical distribution on permutations given by \eqref{finite-box}; an analogous result is obtained in the grand-canonical case. In \S\ref{pm1} we show that the $n$-point correlations of the point marginal of the loop soup match those of the boson point process \eqref{loop-correlations}, a permanental process.

Consider the Brownian transition density 
\begin{equation}
\label{Brownian}
p_t(x,y)=\frac{1}{(2\pi t)^{d/2}}\exp\Bigl\{-\frac{1}{2t}\|x-y\|^2\Bigr\}, \quad t>0,\quad x,y \in \R^d.
\end{equation}  
 Given $k\ge 1$ and $\lambda>0$, we consider the set of \emph{rooted loops of length} $k$, or rooted $k$-loops for short, 
\begin{equation}
\label{rooted-k}
{\DD }^{\rot}_k=\big\{\gamma:\{0,\dots, k\}\to \R^d, \gamma(0)=\gamma(k)\big\}
\end{equation}
 with the Borel $\sigma$-algebra of $(\R^d)^{k+1}$, and
define  $\muq^{\rot}_{k,\lambda}$, the \emph{rooted Gaussian $k$-loop measure with fugacity $\lambda$}, on ${D }^{\rot}_k$, as
\begin{equation}
\label{rooted-bridge}
\muq^{\rot}_{k,\lambda}(d\gamma):=\lambda^k\, \prod_{j=0}^{k-1}p_{\frac{1}{2\alpha}}(x_i, x_{i+1})\,dx_0\dots dx_{k-1}\qquad \text{ if } \gamma(i)=x_i,\quad 0\le i\le k-1.
\end{equation}
We can naturally define the periodic extension of a $k$-loop as a function $\gamma:\Z\to \R^d$ by setting $\gamma(j):=\gamma(m k+j),\, j\in \Z$, where $m\in \Z$ is such that 
$0\le m k+j <k$. Let $\theta^\ell$ denote the time shift by $\ell$, $\theta^\ell(\gamma)(j):=\gamma(j+\ell)$.
An \emph{unrooted loop of length} $k$, or unrooted $k$-loop, is an equivalence class of rooted $k$-loops under the equivalence $\gamma\sim \gamma'$ if there exists $\ell\in \Z$ 
such that $\gamma'=\theta^\ell\gamma$. Let 
\begin{equation}
\label{unrooted-k}
{\DD }_k:={\DD }^{\rot}_k\,/\sim
\end{equation}
denote the set of unrooted $k$-loops and 
\begin{equation}
\label{project}
\pi_k:{\DD }_k^{\rot}\to {\DD }_k
\end{equation}
the projection. Consider the $\sigma$-algebra ${\crD}_k$ on ${\DD }_k$ given by the sets $\fF \subset {\DD }_k$ such that $\pi_k^{-1}(\fF) \subset {\DD }_k^{\rot}$ is measurable, and denote  $\muq_{k,\lambda}$ the induced 
\emph{unrooted Gaussian $k$-loop measure with fugacity $\lambda$} on ${\DD }_k$ defined by 
\begin{equation}
\label{unrooted-bridge}
\muq_{k,\lambda}(\fF):=\frac{1}{k}\,\muq_{k,\lambda}^{\rot}\big(\pi_k^{-1}(\fF)\big).
\end{equation}
The factor $1/k$ in front of the push-forward measure compensates for the fact that a loop class is counted $k$ times by $\muq_{k,\lambda}^{\rot}$, according to each possible root choice. 
For a non negative measurable function $g:{\DD }_k\to \R$, definitions \eqref{unrooted-bridge} and \eqref{rooted-bridge} yield
\begin{align}
\label{k-integral}
\int_{{\DD }_k} g(\gamma) \,\muq_{k, \lambda} (d\gamma)
&=\frac{\lambda^k}{k} \int_{({\R^d})^k} g([x_0,\dots,x_{k-1}])\prod_{i=0}^{k-1} p_{\frac{1}{2\alpha}}(x_i,x_{i+1})\,dx_0\dots dx_{k-1}.
\end{align}
We use the letter $\gamma$ to denote an unrooted loop $[x_0,\dots,x_{k-1}]$. We will abuse notation and define the \emph{support} of $\gamma$ as the set 
$\{\gamma\}:=\{x_0,\dots,x_{k-1}\}\subset \R^d$, and consider $\gamma:\{\gamma\}\to\{\gamma\}$ as a permutation of $\{\gamma\}$ satisfying $\gamma(x_i) = x_{i+1}$. 
Given a Borel set $\Lambda \subset \R^d$ and an unrooted loop $\gamma \in {\DD }_k$ with support $\{\gamma\}=\{x_0,\dots,x_{k-1}\}$, consider the cardinality of the set $\Lambda \cap \{\gamma\}$
\begin{align}
\label{nA}
n_\Lambda(\gamma):= \sum_{j=0}^{k-1} {\mathbf 1}_\Lambda(x_j);\qquad n(\gamma) :=n_{\R^d}(\gamma)=k.
\end{align}
 If the set has finite Lebesgue measure $|\Lambda|<\infty$, 
the mean density of points belonging to $k$-loops in $\Lambda$ is defined as
\begin{equation}
\label{rhok}
\varrho_{k, \lambda}(\Lambda):=\frac{1}{|\Lambda|}\int_{{\DD }_k} n_\Lambda(\gamma)\,\muq_{k, \lambda}(d\gamma)\,.
\end{equation}
\begin{Proposition}[Point density]
\label{densityk}
For any compact set $\Lambda \subset \R^d$ with $|\Lambda| > 0$, we have
\begin{equation}
\label{loop density}
\varrho_{k,\lambda}(\Lambda)=\Big(\frac{\alpha}{\pi}\Big)^{d/2} \frac{\lambda^k}{k^{d/2}}\,.
\end{equation}
\end{Proposition}
\begin{proof}
By \eqref{k-integral} we have
\begin{align*}
\varrho_{k, \lambda}(\Lambda)&=\frac{\lambda^k}{k|\Lambda|} \int_{({\R^d})^k} \Big(\sum_{j=0}^{k-1} {\mathbf 1}_{\Lambda}(x_j)\Big) \prod_{i=0}^{k-1} p_{\frac{1}{2\alpha}}(x_i,x_{i+1})\,dx_0\dots dx_{k-1}\\
&=\frac{\lambda^k}{k|\Lambda|}\,\sum_{j=0}^{k-1}\, \int_{({\R^d})^k} {\mathbf 1}_{\Lambda}(x_j)\, \prod_{i=0}^{k-1} p_{\frac{1}{2\alpha}}(x_i,x_{i+1})\,dx_0\dots dx_{k-1}\\
&=\frac{\lambda^k}{k|\Lambda|}\,\sum_{j=0}^{k-1}\, \int_{\R^d} {\mathbf 1}_{\Lambda}(x_j)\,\Big(\frac{\alpha}{\pi k}\Big)^{d/2}\, dx_j\\
&=\frac{\lambda^k}{k|\Lambda|}\,\Big(\frac{\alpha}{\pi k}\Big)^{d/2}\,k|\Lambda| =\Big(\frac{\alpha}{\pi}\Big)^{d/2} \frac{\lambda^k}{k^{d/2}}\,. \qedhere
\end{align*}
\end{proof}
Consider a compact set $\Lambda$ with positive Lebesgue measure, call $\rho_{k,\lambda}:= \varrho_{k, \lambda}(\Lambda)$ the point density of the measure $\muq_{k, \lambda}$, which does not depend on the particular choice of $\Lambda$, and  
define
\begin{align}
\label{r16}
  \rho(\lambda):= \sum_{k\ge 1}\rho_{k,\lambda}\,.
\end{align}
We have
\begin{equation}
\label{fd}
\rho(\lambda) = \Big(\frac{\alpha}{\pi}\Big)^{d/2} \sum_{k\ge 1} \frac{\lambda^k}{k^{d/2}}<\infty \iff 
\begin{cases} 
d\le 2 \text{ and } 0\le \lambda<1, \hbox{ or}\\ 
d\ge 3 \text{ and } 0\le \lambda \le 1\,.
\end{cases}
\end{equation}
Define the critical density $\rho_c$ by
\begin{align}
\label{rcrit}
  \rho_c := \Big(\frac{\alpha}{\pi}\Big)^{d/2} \sum_{k\ge 1} \frac{1}{k^{d/2}}.
\end{align}

The space $\DD$ of finite unrooted loops in $\R^d$ is 
\begin{align}
\label{r19}
\DD:=\bigcup_{k\ge 1} \DD_k,
\end{align}
with the minimal $\sigma$-algebra ${\crD}$ that contains $\bigcup_{k\ge 1}{\crD}_k$.

If $\lambda$ and $d$ satisfy \eqref{fd}, so that $\rho(\lambda)<\infty$, define the \emph{Gaussian loop soup intensity measure} $\muq^{\ls}_\lambda$ on ${\DD }$ as 
\begin{equation}
\label{loop meas}
\muq_\lambda^{\ls}:=\sum_{k\ge 1} \muq_{k, \lambda}\,.
\end{equation}
This measure has point density $\rho(\lambda)$.

\begin{Proposition}
\label{sigma-finite}
Let  $d$ and $\lambda$ be as in \eqref{fd}. Then  $\muq_\lambda^{\ls}$ is $\sigma$-finite. 
\end{Proposition}
\begin{proof}
We start by showing that for $k\in \N$ and $\lambda>0$, the measure $\muq_{k, \lambda}$ on ${\DD }_k$ is $\sigma$-finite.
Let $B_j\subset \R^d$ be the ball of radius $j$ in $\R^d$, and define
$
\Upsilon_{k,j}:=\big\{\gamma\in {\DD }_k: \{\gamma\}\cap B_j \neq\emptyset\big\}\,.
$
Then $\bigcup_j \Upsilon_{k,j}={\DD }_k$, and
\begin{equation*}
\muq_{k, \lambda}(\Upsilon_{k,j})\le |B_j| \,\varrho_{k,\lambda}(B_j)=|B_j| \rho_{k, \lambda} <\infty,
\end{equation*}
by \eqref{rhok} and Proposition \ref{densityk}.
Letting  
$\Upsilon_j:=\bigcup_{k\ge 1} \Upsilon_{k,j}$,
we have $\bigcup_{j\ge 1}  \Upsilon_j={\DD }$. Then
\[
\muq^{\ls}_{\lambda}( \Upsilon_j)\le \sum_{k\ge 1} \muq_{k,\lambda}( \Upsilon_{k,j})\le \sum |B_j| \rho_{k, \lambda}=\rho |B_j|<\infty,\quad\hbox{by \eqref{fd}.}\qedhere
\] 
\end{proof}

\begin{Definition}[Gaussian loop soup]
\label{loop soup}\rm
Let $d$ and $\lambda$ satisfy \eqref{fd}. We define the \emph{Gaussian loop soup} at fugacity $\lambda$ as a Poisson  process on $\DD $ with intensity measure $\muq^{\ls}_{\lambda}$. 
We will use $\Gamma^{\ls}_\lambda$ to denote a realization of the process, and $\mu^{\ls}_\lambda$ to denote its law. 
\end{Definition}
A configuration $\Gamma$ of a Gaussian loop soup is a countable collection of unrooted Gaussian loops in $\R^d$, with the property that any compact set contains finitely many points in the supports of these loops.
Let $\kD$ be the set of locally finite loop soup configurations,
\begin{align}
  \label{xc21}
  \kD:=\{\Gamma\subset \DD: \tsum_{\gamma\in\Gamma}\;
  n_\Lambda(\gamma) <\infty, \hbox{ for all compact }\Lambda\subset\R^d\},
\end{align}
where $n_\Lambda(\gamma)$ is defined in \eqref{nA}.

Let $\fF\subset \DD$ have finite Gaussian loop soup intensity measure, $Q^{\ls}_\lambda(\fF)<\infty$, and consider a measurable, bounded function $g:\kD\to\R$ that is determined by the loops of the configuration that belong to $\fF$, i.e. $g(\Gamma)=g(\Gamma')$ if $\Gamma\cap \fF=\Gamma'\cap \fF$.
By definition of Poisson process, we can explicitly write
\begin{align}
  \mu^{\ls}_\lambda g
  &= e^{-Q^{\ls}_\lambda(\fF)}\sum_{\ell\ge0}\frac{1}{\ell!}
  \int_{\fF}\dots\int_{\fF}
  g\big(\{\gamma_1,\dots,\gamma_\ell\}\big) \,
  Q^{\ls}_\lambda(d\gamma_1)\dots Q^{\ls}_\lambda(d\gamma_\ell) \label{fls30}\\
  &= \sum_{\ell\ge0}\frac{e^{-Q^{\ls}_\lambda(\fF)}}{\ell!}
  \int_{\fF}\dots\int_{\fF}
  g\big(\{\gamma_1,\dots,\gamma_\ell\}\big) \,
  f^{\ls}_\lambda\big(\{\gamma_1,\dots,\gamma_\ell\}\big)\,
  d\gamma_1\dots d\gamma_\ell, \label{fls22}
\end{align}
where
\begin{align}
  f^{\ls}_\lambda\big(\{\gamma_1,\dots,\gamma_\ell\}\big)
  &:= 
    \textstyle{\prod_{i=1}^\ell}\, \omega_\lambda(\gamma_i),\label{fls1}\\[2mm]
  \omega_\lambda\big([x_0,\dots,x_{k-1}]\big)
  &:= \lambda^k\, \textstyle{\prod_{i=0}^{k-1}}\, p_{\frac1{2\alpha}}(x_i,x_{i+1})
    \quad \hbox{with }x_{k}=x_0,\label{fls2}
\end{align}
and where for any bounded measurable $h:\fF\to \R$,
\begin{align}
  \label{eq:7}
  \int_{\DD} h(\gamma)\,d\gamma
  := \sum_{k\ge 1}\frac{1}{k}\,\int_{\R^d}\dots \int_{\R^d}
  h\big([x_0,\dots,x_{k-1}]\big)\, dx_0\dots dx_{k-1},
\end{align}
if the right hand side is well defined. Taking  $h(\gamma)= h_1(\gamma)\,\omega_\lambda(\gamma)\one_{\{\gamma\in\fF\}}$ for some bounded $h_1:\fF\to\R$ and recalling the assumption $Q^{\ls}_\lambda(\fF)<\infty$, we conclude that  \eqref{eq:7} is well defined for this $h$ and it is bounded by $\|h_1\|_\infty Q^{\ls}_\lambda(\fF)$.

Given $\fF\subset \DD$ 
denote the set of loop configurations contained in $\fF$ by
\begin{align}
  \label{xc22}
  \kD_\fF:= \big\{  \Gamma \in \kD:\,\Gamma\subseteq \fF  \big\}.
\end{align}
Denote $F^{\ls}_\lambda(\Gamma):= e^{-Q^{\ls}_\lambda(\fF)}f^{\ls}_\lambda(\Gamma)$. Given a bounded test function $g:\kD_\fF\to\R$, we get
\begin{align}
  \int\mu^{\ls}_\lambda(d\Gamma) g(\Gamma\cap\fF) &= \int_{\kD_\fF} g(\Gamma)\,F^{\ls}_\lambda(\Gamma)\, d\Gamma\label{dGamma}\\
  &=  \sum_{\ell\ge0} e^{-Q^{\ls}_\lambda(\fF)}
  \int_{\fF}\dots\int_{\fF}
  g\big(\{\gamma_1,\dots,\gamma_\ell\}\big) \,\one_{\gamma_1<\cdots<\gamma_\ell}\;
    f^{\ls}_\lambda\big(\{\gamma_1,\dots,\gamma_\ell\}\big)\,
  d\gamma_1\dots d\gamma_\ell,  \notag
\end{align}
which is equivalent to \eqref{fls22}.
In the above equation we consider the total order in $\DD$ induced by the lexicographic order in $\R^d$: $\gamma<\gamma'$ if $\min\{\gamma\}<\min\{\gamma'\}$. Definitions \eqref{fls22} and \eqref{dGamma} are equivalent because $f^{\ls}_\lambda$ and $g$ do not depend on the order. The function $F^{\ls}_\lambda(\Gamma)$  is the \emph{density} of the Gaussian loop soup at fugacity $\lambda$ in the set $\fF$.

\subsection{Finite-volume loop soup and spatial permutations}
\label{fls5}
In this subsection we show that the spatial random permutation with density \eqref{finite-box} has the same law as the Gaussian loop soup in the bounded domain $\Lambda$, conditioned to have $N$ points; this measure is well defined as the conditioning set has positive probability. Similarly, we prove that the grand-canonical measure associated to  \eqref{finite-box} is the loop soup in $\Lambda$. 

Note that there is 
  a bijection between $\kD$ and the set
  \begin{align}
    \label{eq:2}
    \{(\sX,\sigma): \sX \text{ locally finite,}  \hbox{ for all }x\in\sX\, \text{ there is } k<\infty \hbox{ with } \sigma^k(x) =x\},
  \end{align}
  the space of finite cycle permutations with locally finite supports.
Indeed, for any $\Gamma\in\kD$, let $(\sX, \sigma)$ be given by $\sX = \cup_{\gamma\in\Gamma} \{\gamma\}$ and $\sigma(x)=\gamma(x)$, $x\in \{\gamma\},\,\gamma \in \Gamma$. 
Conversely, given a spatial permutation $(\sX,\sigma)$ in the set \eqref{eq:2}, define the loop soup  configuration
\begin{align*}
\Gamma := \cup_{x\in\sX} \{[x,\sigma(x),\dots,\sigma^{k(x)-1}(x)]\},
\end{align*}
where  $k(x):=\min\{j>0:\sigma^j(x)=x\}$ is the size of the cycle in $\sigma$ containing $x$. 
Define the sets
\begin{align}
  \DD_\Lambda
  &:=\{\gamma\in \DD:\big\{\gamma\}\subset\Lambda\big\},&&\text{set of loops with supports contained in $\Lambda$},
    \notag\\
     \kD_{\Lambda}
    &:= \{\Gamma\in\kD: \Gamma\subset \DD_\Lambda\}, \label{dl28}&&\text{loop configurations contained in $\Lambda$},\\
   \kD_{\Lambda,N}&:= \big\{\Gamma\in\kD_{\Lambda}:
    \textstyle{\sum_{\gamma\in\Gamma}} n(\gamma)=N\big\},&&\text{loop configurations contained in $\Lambda$ with $N$ points}.\notag
\end{align}

\paragraph{Canonical measures}
The Gaussian loop soup restricted to loops contained in $\Lambda$ is defined by
\begin{align}
  \mu^{\ls}_{\Lambda,\lambda}
  &:=
    \text{Poisson process on }\kD_\Lambda
    \text{ with intensity measure }\;
    \one_{\DD_\Lambda}(\gamma) \,Q^{\ls}_{\lambda}(d\gamma).
    \label{23a}
\end{align}
We now show that $\mu^{\ls}_{\Lambda,\lambda}$ conditioned to have $N$ points in $\Lambda$ equals the spatial random permutation with density \eqref{finite-box}, that we denote 
$\mu_{\Lambda,N}$.
\begin{Proposition}[Canonical measure and loop soup] For any measurable, bounded test function $g:\kD_{\Lambda,N} \to \R$, we have 
  \begin{align}
  \label{31z}
    \frac{1}{\mu^{\ls}_{\Lambda,\lambda}(\kD_{\Lambda,N})}\int_{\kD_{\Lambda,N}}g(\Gamma)\, {\mu^{\ls}_{\Lambda,\lambda}(d\Gamma)} = \mu_{\Lambda,N}g.
  \end{align}
\end{Proposition}

\begin{proof}
  Since there is a bijection between the supports of these probability measures, it suffices to verify that the weights assigned by their densities to any given configuration satisfy a fixed ratio. Let $(\sX,\sigma)$ be a spatial permutation such that $\sX\subset \Lambda$ and $n_\Lambda(\sX)=N$, where $n_\Lambda(\sX)$ is the cardinality of $\sX$. Let $\Gamma$ be the cycle decomposition of $(\sX,\sigma)$; clearly $\Gamma\in \kD_{\Lambda,N}$. Then, by \eqref{dGamma}, the loop soup conditioned density of $\Gamma\in\kD_{\Lambda,N}$ is
  \begin{align}
    \label{23b}
    F^{\ls}_{\lambda}(\Gamma|\kD_{\Lambda,N})
    = \frac{e^{-Q^{\ls}_\lambda(\DD_\Lambda)}}
    {\mu^{\ls}_{\lambda}(\kD_{\Lambda,N})}
    \prod_{\gamma\in\Gamma} \omega_\lambda(\gamma)
    = \frac{\lambda^N e^{-Q^{\ls}_\lambda(\DD_\Lambda)}}
    {\mu^{\ls}_{\lambda}(\kD_{\Lambda,N})}
    \prod_{\gamma\in\Gamma} \omega_1(\gamma),  
  \end{align}
  where $\omega_\lambda$ was defined in \eqref{fls2}. 
  On the other hand, the density \eqref{finite-box} of the canonical measure $\mu_{\Lambda,N}$ can be written as a function of the cycle decomposition of $\sigma$ by 
  \begin{align}
    \label{23bb}
    F_{\Lambda,N}(\sX,\sigma)
    &= \frac{(\alpha/\pi)^{Nd/2}}{Z_{\Lambda,N}}\,
   e^{-\alpha H(\sX, \sigma)}
    = \frac{1}{Z_{\Lambda,N}}
    \prod_{\gamma\in (\sX,\sigma)} \omega_1(\gamma). 
  \end{align}
  Since both \eqref{23b} and \eqref{23bb} determine probability measures, the normalization factors must be the same. 
\end{proof}

\paragraph{Grand-canonical measures}
The grand-canonical spatial random permutation at fugacity $\lambda< 1$ if $d=1,\,2$, or $\lambda\le 1$ if $d\ge 3$, associated to the canonical density \eqref{finite-box}, is defined by
\begin{align}
  \mu_{\Lambda,\lambda}\,g
  &:= \frac{1}{Z_{\Lambda,\lambda}}\sum_{N\ge 0}\lambda^N\, \frac{(\alpha/\pi)^{Nd/2}}{N!}
    \sum_{\sigma\in\cS_N}
    \int_{\Lambda^N}
    g(\ux,\sigma)\;
    e^{-\alpha H(\ux,\sigma )} d\ux,  \label{23c}
\end{align}
where $\ux=(x_1,\dots,x_N)$, $H(\ux,\sigma) := \sum_{i=1}^N \|x_i-x_{\sigma(i)}\|^2$ and
\begin{align}\label{233c}
  Z_{\Lambda,\lambda}:=\sum_{N\ge 0} \lambda^N\,\frac{(\alpha/\pi)^{Nd/2}}{N!}
    \sum_{\sigma\in\cS_N}
    \int_{\Lambda^N}
    e^{-\alpha H(\ux,\sigma )} d\ux \;=\;\sum_{N\ge 0}\,\lambda^N\,Z_{\Lambda,N} .
\end{align}
Elementary computations show that $Z_{\Lambda,\lambda}<\infty$ for $\lambda$ in the specified domains. We next show the equivalence of the loop soup and the grand-canonical measure.
\begin{Proposition}[Grand canonical measure and loop soup]
  \label{p24}
Let $\lambda< 1$ if $d=1,\,2$, or $\lambda\le 1$ if $d\ge 3$. The Gaussian loop soup at fugacity $\lambda$ restricted to $\Lambda$ defined in \eqref{23a} and the grand-canonical measure \eqref{23c} at the same fugacity are equal,
  \begin{align}
\label{36j}
    \mu^{\ls}_{\Lambda,\lambda} = \mu_{\Lambda,\lambda}.
  \end{align}
\end{Proposition}
\begin{proof}
By \eqref{23c} and \eqref{233c}, the grand canonical measure is a mixture of canonical measures: 
  \begin{align}
  \label{gcmix}
    \mu_{\Lambda,\lambda}
    &= \sum_{N\ge0}  \;\frac{\lambda^N\,Z_{\Lambda,N}}{Z_{\Lambda,\lambda}}\;\mu_{\Lambda,N}.
  \end{align}
By \eqref{31z}, the loop soup restricted to $\Lambda$ is also a mixture of canonical measures:
 \begin{align}
\label{lsmix}
    \mu^{\ls}_{\Lambda,\lambda}
    &= \sum_{N\ge0}\; {\mu^{\ls}_{\lambda}(\kD_{\Lambda,N})}\; \mu_{\Lambda,N}.
 \end{align}
 The identity between \eqref{23b} and \eqref{23bb} implies that
 \begin{align}
   \mu^{\ls}_{\lambda}(\kD_{\Lambda,N}) = \lambda^N \,Z_{\Lambda,N}\;e^{-Q^{\ls}_\lambda(\DD_\Lambda)}.
 \end{align}
We add over $N$ and apply \eqref{233c} to get $Z_{\Lambda,\lambda}= e^{Q^{\ls}_\lambda(\DD_\Lambda)}$, so that the weights in \eqref{gcmix} and \eqref{lsmix} are equal.
\end{proof}

\subsection{Point marginal of the Gaussian loop soup}
\label{pm1}
In this subsection we study the law of the point marginal of the Gaussian loop soup. 
Let $\Gamma^{\ls}_\lambda$ be the loop soup with distribution $\mu^{\ls}_\lambda$ and denote $\sX^{\ls}_\lambda$ the set of points in the support of $\Gamma^{\ls}_\lambda$:
\begin{equation}
\label{loop points}
\sX^{\ls}_\lambda := \bigcup_{\gamma\in\Gamma^{\ls}_\lambda} \{\gamma\},\qquad \nu^{\ls}_\lambda := \hbox{ law of }\sX^{\ls}_\lambda.
\end{equation}
For $k\ge 1$, $\lambda>0$, define
\begin{align}
&J^k_\lambda(x,y):=\lambda^k \int_{(\R^d)^{k-1}} p_{\frac{1}{2\alpha}}(x,z_1) \dots  p_{\frac{1}{2\alpha}}(z_{k-1}, y)\, d\uz
=\lambda^k\Big(\frac{\alpha}{\pi k}\Big)^{d/2}\exp\Bigl\{-\frac{\alpha}{k}\|x-y\|^2\Bigr\}, \label{Jk}
\end{align}
where $\uz=(z_1,\dots,z_{k-1})$.  With $d$ and $\lambda$ as in \eqref{fd}, $K_\lambda(x,y)$ defined in \eqref{Klambda} satisfies $K_\lambda(x,y)=\sum_{k\ge 1} J^k_\lambda(x,y)<\infty$. 

Let $\nu$ be the distribution of a point process on $\R^d$ and $\fh:\R^d\to\R_{\ge0}$ be a non-negative measurable function with compact support. Recall that $\cL(\fh):= \int \nu(d\sX) \exp \big(- \sum_{x\in\sX} \fh(x) \big)$ is the Laplace functional of the point process with probability distribution $\nu$.  Denote  
 \begin{equation}
 \label{parti}
 {\cP}_n
 :=\bigl\{ P \text{ partition of }\{1,\dots,N\}\;:\; \emptyset\notin P\bigr\}.
\end{equation}

\begin{Proposition}[Laplace functionals]
\label{finite-Laplace}
 Let $\lambda< 1$ if $d=1,\,2$, or $\lambda\le 1$ if $d\ge 3$.
Let $\fh:\R^d\to \R$ be a non-negative measurable function with compact support. Denote $\cL^{\ls}_\lambda$ the Laplace functional of the point marginal of the Gaussian loop soup at fugacity $\lambda$. Then, if $\phi$ satisfies  $K_\lambda(0,0)\|e^{-\phi}-1\|_1<1$ we have
\begin{align}
  \cL^{\ls}_\lambda(\fh)&=\exp\Big\{\sum_{j\ge 1} \frac{1}{j}\int_{(\R^d)^j}\prod_{i=0}^{j-1} (e^{-\fh(y_i)}-1)\,K_\lambda(y_i,y_{i+1})\,dy_0\dots dy_{j-1}\Big\}\label{F1}\\[2mm]
&=1+\sum_{n\ge 1} \frac{1}{n!}\int_{(\R^d)^n} \sum_{\sigma \in \cS_n} \prod_{i=0}^{n-1} \big(e^{-\fh(y_i)}-1\big)\,\prod_{i=0}^{n-1} K_\lambda(y_i,y_{\sigma(i)})\,dy_0\dots dy_{n-1},\label{F2}
\end{align}
with the convention that $y_j=y_0$ and $y_n=y_0$ in the integrals over $(\R^d)^j$ and $(\R^d)^n$, respectively,  on the right of \eqref{F1} and \eqref{F2}.
\end{Proposition}
\begin{proof}
Let $\Fh(\gamma):=\sum_{x\in \{\gamma\}} \fh(x)$, $\gamma \in {\DD }$, so that we can write
\begin{align}
\label{laplaceP}
  \cL^{\ls}_\lambda(\phi) &= \int \nu^{\ls}_\lambda(d\sX)\,\exp\Bigl\{-\sum_{x\in \sX}\fh(x)\Bigr\}=\int \mu^{\ls}_\lambda(d\Gamma)\,\exp\Bigl\{-\sum_{\gamma\in \Gamma}\Fh(\gamma)\Bigr\}\nn\\
&=\exp\Big\{ \int_{\DD } (e^{-\Fh(\gamma)}-1) \,\muq_{\lambda}(d\gamma)\Big\},
\end{align}
by Campbell's theorem.
Notice that $e^{-\Fh(\gamma)}-1$ is bounded and vanishing in the complement of  $\fF_j$ defined in Proposition \ref{sigma-finite}, if $j\in \N$ is such that the support $\text{supp}(\fh)\subseteq B_j$, the ball of radius $j$ in $\R^d$. We saw in the proof of Proposition \ref{sigma-finite} that $\muq_\lambda^{\ls}(\fF_j)$ is finite, hence the integral in the exponent of \eqref{laplaceP} converges.

If $\gamma=[x_0,\dots,x_{k-1}] \in {\DD }_k$ then 
\begin{align*}
e^{-\Fh(\gamma)}-1&=\prod_{i=0}^{k-1} \big[\big(e^{-\fh(x_i)}-1\big)+1\big]\,-1=\sum_{\emptyset\neq \gamma'\subset \gamma} \;\;\prod_{x\in\{\gamma'\}} \big(e^{-\fh(x)}-1\big),
\end{align*}
where $\gamma'\subset \gamma$ means that the trace of the loop $\gamma'$ is a subset of the trace of $\gamma$: $\gamma'\in {\DD }_j$ with $j\le k$, and
there is  a subsequence $\{n_i\}_{0\le i\le j-1}$ of $0,\dots, k-1$ such that 
$\gamma' = [x_{n_0},\dots,x_{n_{j-1}}]$.

Then
\begin{align}
 &\int_{\DD } (e^{-\Fh(\gamma)}-1) \,\muq_{\lambda}(d\gamma)\notag\\
&\qquad=\sum_{k\ge 1}\int_{{\DD }_k} (e^{-\Fh(\gamma)}-1) \,\muq_{k,\lambda}(d\gamma) \notag\\
  &\qquad=\sum_{k\ge 1}\int_{{\DD }_k} \; \sum_{\emptyset\neq \gamma'\subset \gamma}\;\;
    \prod_{x\in\{\gamma'\}} \big(e^{-\fh(x)}-1\big)\;\muq_{k,\lambda}(d\gamma)\, \notag\\
&\qquad=\sum_{k\ge 1}\frac{\lambda^k}{k}\int_{(\R^d)^k}  \hspace{-.5cm}\sum_{\substack{j\le k \\ \{n_0,\dots , n_{j-1}\}\\ \ \ \subset\{0,\dots , k-1\} }}
\prod_{i=0}^{j-1} \big(e^{-\fh(x_{n_i})}-1\big) \,\, \prod_{i=0}^{k-1} p_{\frac{1}{2\alpha}}(x_i,x_{i+1})\,dx_0\dots dx_{k-1} \label{red1}\\
&\qquad=\sum_{k\ge 1} \sum_{j\le k} \frac{1}{j}  \hspace{-3mm}\sum_{\substack{r_0\ge 1\\ \dots \\ r_{j-1}\ge 1\\ r_0+\dots+r_{j-1}=k}}\hspace{-2mm}
\int_{(\R^d)^j} \prod_{i=0}^{j-1} \big(e^{-\fh(y_i)}-1 \big)\,\prod_{i=0}^{j-1} J_\lambda^{r_i}(y_i,y_{i+1})\,dy_0 \dots dy_{j-1} \label{red3}\\
&\qquad= \sum_{j\ge 1} \frac{1}{j} \sum_{\substack{r_0\ge 1\\ \dots \\ r_{j-1}\ge 1}} \int_{(\R^d)^j} \prod_{i=0}^{j-1} \big(e^{-\fh(y_i)}-1 \big)\,\prod_{i=0}^{j-1} J_\lambda^{r_i}(y_i,y_{i+1})\,dy_0 \dots dy_{j-1}\,.\label{red2}
\end{align}
In order to derive \eqref{red2} we first fix the points $y_0,\dots,y_{j-1} \in \gamma'$, i.e. the arguments of the factors $e^{-\fh(\cdot)}-1$ in \eqref{red1}, and integrate over the remaining points of those loops $\gamma \in {\DD }_k,\, k\ge j$, such that $\gamma'\subset \gamma$. This leads to the sum over the family of indices  $(r_0,\dots, r_{j-1})\in \N^j$ and the product of transition kernels $\prod_i J_\lambda^{r_i}(y_i, y_{i+1})$ in the integral over $\big(\R^d \big)^j$. Note that the integral in \eqref{red3} assigns density 
\begin{equation*}
j\,\prod_{i=0}^{j-1} \big(e^{-\fh(y_i)}-1 \big)\,\prod_{i=0}^{j-1} J_\lambda^{r_i}(y_i,y_{i+1})
\end{equation*}
to an unrooted loop 
$\gamma'=[y_0,\dots,y_{j-1}]$ and a choice of indices $(r_0,\dots,r_{j-1})$, which is corrected by the factor $1/j$ in front to match the corresponding expression in \eqref{red1}. Now \eqref{F1} follows by changing the order of the sum and the integral. 
This works if there is absolute convergence of the exponent in \eqref{F1}; using $K_\lambda(x_i,x_{i+1}) \leq K_\lambda(0,0) = 
  \sum_{k \geq 1} \bigl(\frac{\alpha}{\pi k} \bigr)^{\frac{d}{2}} \lambda^k$
  and $\ K_\lambda(0,0) \| e^{-\phi} - 1 \|_1 < 1$, we have the absolute convergence. 
 
 In order to prove \eqref{F2}, write
 \begin{align*}
&\frac{1}{j}\int_{(\R^d)^j} \prod_{i=0}^{j-1} \big(e^{-\fh(y_i)}-1\big)\, K_\lambda(y_i,y_{i+1})\,dy_0\dots dy_{j-1}\\
&\hspace{2cm}=\frac{1}{j!} \int_{(\R^d)^j}\prod_{i=0}^{j-1} \big(e^{-\fh(y_i)}-1\big)\,\Big[\sum_{\sigma \in C_j} \prod_{i=0}^{j-1} K_\lambda(y_i,y_{\sigma(i)})\Big]\,dy_0\dots dy_{j-1}.
\end{align*}
where $C_j$ denotes the set of cycles of length $j$, 
\[
C_j\subset \cS_j=\{ \sigma: [0,j-1] \to [0, j-1] , \sigma \text{ bijective} \},
\] 
the set of permutations of $j$ elements. We replace this expression in \eqref{red2} and denote by $|I|$ the cardinality of $I$, to get
\begin{align}
  \cL^{\ls}_\lambda (\fh)&=\exp\Big\{\sum_{ j\ge 1}\frac{1}{j!} \sum_{\sigma \in C_j}\int_{(\R^d)^j}\prod_{i=0}^{j-1} \big(e^{-\fh(y_i)}-1\big)\,\prod_{i=0}^{j-1} K_\lambda(y_i,y_{\sigma(i)})\,dy_0\dots dy_{j-1}\Big\} \notag\\
&=1+\sum_{n\ge 1} \frac{1}{n!} \sum_{P\in \cP_n} \prod_{I\in P} 
\sum_{\sigma \in C_{|I|}}\int_{(\R^d)^{|I|}}\prod_{i=0}^{|I|-1} \big(e^{-\fh(y_i)}-1\big)\,\prod_{i=0}^{|I|-1} K_\lambda(y_i,y_{\sigma(i)})\,dy_0\dots 
dy_{|I|-1}\label{appli}\\
&=1+\sum_{n\ge 1} \frac{1}{n!}\int_{(\R^d)^n} \sum_{\pi \in \cS_n} \prod_{i=0}^n \big(e^{-\fh(y_i)}-1\big)\,\prod_{i=0}^n K_\lambda(y_i,y_{\pi(i)})\,dy_0\dots dy_n, \notag
\end{align}
where in identity \eqref{appli} we have used Lemma~\ref{exp}. This shows \eqref{F2}. \end{proof}

\begin{Remark}
\label{combi}
{\rm Identity \eqref{appli} is used in quantum field theory and statistical mechanics to express the log of the partition function as a sum over connected graphs. We thank a referee for mentioning that a version of the identity is contained in Theorem II.1 of Flajolet and  Sedgewick \cite{zbMATH05485323}, a reference in combinatorics. For the convenience of the reader, we include an elementary proof in our next lemma. 
}
\end{Remark}

\begin{Lemma}[Combinatorial lemma]
\label{exp}
Let $(a_n)_{n\ge 1} \in \Co^{\N}$  be such that $\sum_{n\ge 1} \frac{1}{n!} |a_n| <\infty$. Then 
\begin{equation}
\label{exp1}
\exp\Big(\sum_{n\ge 1}\frac{1}{n!}\,a_n \Big)=1+\sum_{n\ge 1} \frac{1}{n!} \sum_{P\in \cP_n} \prod_{I\in P} a_{|I|},
\end{equation} 
where $\cP_n$ is defined in  \eqref{parti} and $|I|$ stands for the cardinality of the set $I$.
\end{Lemma}
\begin{proof}
By the series expansion of the exponential function
\begin{align*}
 & \exp\Big(\sum_{n\ge 1}\frac{1}{n!}\,a_n \Big)
  =\sum_{j\ge 0} \frac{1}{j!}
   \Big(\sum_{\ell\ge 1} \frac{1}{\ell!} a_l\Big)^j \;
   = \;1+\sum_{j\ge 1} \frac{1}{j!} \sum_{\substack{i_1,\dots,i_j \\ i_\ell\ge 1}} \frac{a_{i_1}}{i_1!}\dots\frac{a_{i_j}}{i_j!}\notag\\
  &\quad=1+\sum_{j\ge 1} \frac{1}{j!}\sum_{n\ge 1}\  \sum_{\substack{i_1+\dots+i_j=n \\ i_\ell \ge 1}}  \frac{a_{i_1}}{i_1!}\dots\frac{a_{i_j}}{i_j!}  \;
  =\;1+\sum_{j\ge 1} \frac{1}{j!}\sum_{n\ge 1}\frac{1}{n!} \sum_{\substack{i_1+\dots+i_j=n \\ i_\ell \ge 1}}  {{n}\choose{i_1\dots i_j}}\  a_{i_1} \dots a_{i_j} \notag\\
&\quad=1+\sum_{n\ge 1}\frac{1}{n!}\sum_{j\ge 1} \frac{1}{j!}\sum_{\substack{i_1+\dots+i_j=n \\ i_\ell \ge 1}}  {{n}\choose{i_1\dots i_j}}\  a_{i_1} \dots a_{i_j} 
\; =\;1+\sum_{n\ge 1}\frac{1}{n!}\sum_{j\ge 1}\ \sum_{\substack{P=\{I_1,\dots, I_j\}\, \in\, {\cP}_n}} a_{|I_1|}\dots a_{|I_j|}, \end{align*}
which is the right hand side of \eqref{exp1}. 
The interchange of $\sum_n$ and $\sum_j$ in the last line is justified by the absolute convergence of $\sum_{n\ge 1}\frac{1}{n!} \,a_n$.
\end{proof}

\subsubsection{Point correlations}
\label{pc-gls}

Given a random point process $\sX$ on $\R^d$, a non-negative, measurable function $\varphi:\big(\R^d\big)^n\to \R$ is the \emph{$n$-point correlation density} of $\sX$, if for any collection of bounded, pairwise disjoint  Borel sets $\Lambda_1,\,\dots,\Lambda_n \subset \R^d$,
\begin{align*}
E\bigl[n_{\Lambda_1}(\sX) \dots n_{\Lambda_n}(\sX) \bigr]=\int_{\Lambda_1\times\dots\times \Lambda_n} \varphi(x_1,\dots,x_n)\,dx_1\dots dx_n.
\end{align*}

\begin{Proposition}[Point correlations]
\label{ls-correl}
  The $n$-point correlation density of $\nu_\lambda^{\ls}$ is given by
\begin{align}
  \label{loop-corr}
  \varphi_\lambda^{\ls}(x_1,\dots, x_n) = \perm\bigl(K_\lambda(x_i,x_j)\bigr)_{i,\,j=1}^n,
  \end{align}
where $\perm(K)$ is the permanent of the matrix $K\in \R^{n\times n}$.
\end{Proposition}
\begin{proof}[Sketch of the proof]

We compute the $3$-point correlation density. 
To simplify notation, in this proof we will denote $\mu=\mu^{\ls}_\lambda$, $Q=Q^{\ls}_\lambda$ and $K_{xy}= K_\lambda(x,y)$. Given pairwise disjoint bounded Borel sets $A,B,C\subset\R^d$, the third moment measure for the point marginal $\nu^{\ls}_\lambda$ over $A\times B\times C$ is given by
\begin{align}
  &\int n_{A}(\Gamma) \,n_{B}(\Gamma)\,n_{C}(\Gamma)\,\mu(d\Gamma)\label{3mim}\\
  &\quad=\int
    \tsum_{\gamma\in\Gamma}\, n_{A}(\gamma)\,
    \tsum_{\gamma'\in\Gamma}\,n_{B}(\gamma')\,
    \tsum_{\gamma''\in\Gamma}\,n_{C}(\gamma'')\,
    \mu(d\Gamma)\nn\\
  &\quad=\int\Bigl(
    \tsum_{\gamma\in\Gamma}\,
    n_{A}(\gamma)\,n_{B}(\gamma)\,n_{C}(\gamma)\,
    + \tsum_{\gamma\in\Gamma}\, n_{A}(\gamma)\, \tsum_{\gamma'\in\Gamma,\,\gamma'\neq\gamma}\,n_{B}(\gamma')\,n_{C}(\gamma')\label{4mim}\\[1mm]
  &\quad\quad\quad+\, \tsum_{\gamma\in\Gamma}\, n_{B}(\gamma) \,\tsum_{\gamma'\in\Gamma,\,\gamma'\neq\gamma}\,n_{A}(\gamma')\,n_{C}(\gamma')
 \, + \,\tsum_{\gamma\in\Gamma}\, n_{C}(\gamma) \,\tsum_{\gamma'\in\Gamma,\,\gamma'\neq\gamma}\,n_{A}(\gamma')\,n_{B}(\gamma')\nn\\[1mm]
  &\quad\quad\quad+\, \tsum_{\gamma\in\Gamma}\, n_{A}(\gamma)\, \tsum_{\gamma'\in\Gamma,\,\gamma'\neq\gamma}\,n_{B}(\gamma')\,\tsum_{\gamma''\in\Gamma\setminus\{\gamma,\gamma'\}}\,n_{C}(\gamma'')
    \Bigr)\mu(d\Gamma)\nn\\[1mm]
  &\quad= Q( n_{A}\,n_{B}\,n_{C})\,
    + Q(n_{A})\,
    Q(n_{B}\,n_{C})\label{5mim}\\[2mm]
  &\quad\quad
    + Q(n_{B})\,
    Q(n_{A}\,n_{B})
    + Q(n_{C})\,
    Q(n_{A}\,n_{B})
    + Q(n_{A})\,
    Q(n_{B})\,Q(n_{C}).
   \nn
\end{align}
To go from \eqref{4mim} to \eqref{5mim} we use that if $\mu$ is a Poisson process of loops, then (a) the expectation of the product of functions of different loops factorizes (Theorem 3.2 in \cite{MR2004226}), and (b) $\int \sum_{\gamma\in\Gamma} g(\gamma)\mu(d\Gamma) = \int_\DD g(\gamma) Q(d\gamma)$, denoted $Q(g)$, by Campbell's theorem.

Define
\begin{align}
  \langle a_1\dots a_k\rangle
  & :=\; \big\{\gamma\in \DD : \gamma \text{ goes through }a_1,\dots, a_k \text{ in this order}\big\}\label{lo45}
\end{align}
and compute
\begin{align}
  \label{ar5}
  Q( n_{A}\,n_{B}\,n_{C})
  &= \int\tsum_{a\in \{\gamma\}}\one_{A}(a)\,\tsum_{b\in \gamma}\one_{B}(b)\,\tsum_{c\in \gamma}\one_{C}(c) \,Q(d\gamma)\\
  &= \int
    \sum_{\substack{\{a,b,c\}\subset\{\gamma\}
  \\ (a,b,c)\in A\times B\times C}}
  \big(\one_{\langle abc\rangle}(\gamma)
  + \one_{\langle acb\rangle}(\gamma)\big)\,
  Q(d\gamma)\label{ar6} \\[2mm]
  &= \int_A\int_B\int_C (K_{ab}K_{bc}K_{ca} + K_{ac}K_{cb}K_{ba})\, dc\,db\,da,\label{ar7}
\end{align}
where \eqref{ar6} follows from partitioning the set of cycles that go through $a,\,b,\,c$ according to the order in which they visit the points, and
 \eqref{ar7}  can be proved using the argument applied to derive \eqref{F1} from \eqref{red1}. 
Using the same argument to compute the other terms in \eqref{5mim}, we conclude that the  third moment measure \eqref{3mim} is absolutely continuous  with respect to Lebesgue measure in $(\R^d)^3$ with Radon-Nikodym derivative 
\begin{align}
  \varphi_\lambda^{\ls}(x,y,z) &=
  K_{xx}K_{yy}K_{zz} + K_{xx}K_{yz}K_{zy} + K_{xy}K_{yx}K_{zz} \nn\\
                 &\qquad + K_{xy}K_{yz}K_{zx}+ K_{xz}K_{yx}K_{zy}+ K_{xz}K_{yy}K_{zx},\nn
\end{align}
which proves \eqref{loop-corr} for $n=3$; see Fig.~\ref{correlations-1}. We leave the proof of the general case to the reader.
\end{proof}

   \noindent\begin{minipage}{\linewidth}\begin{center}
    \includegraphics[scale=.5]{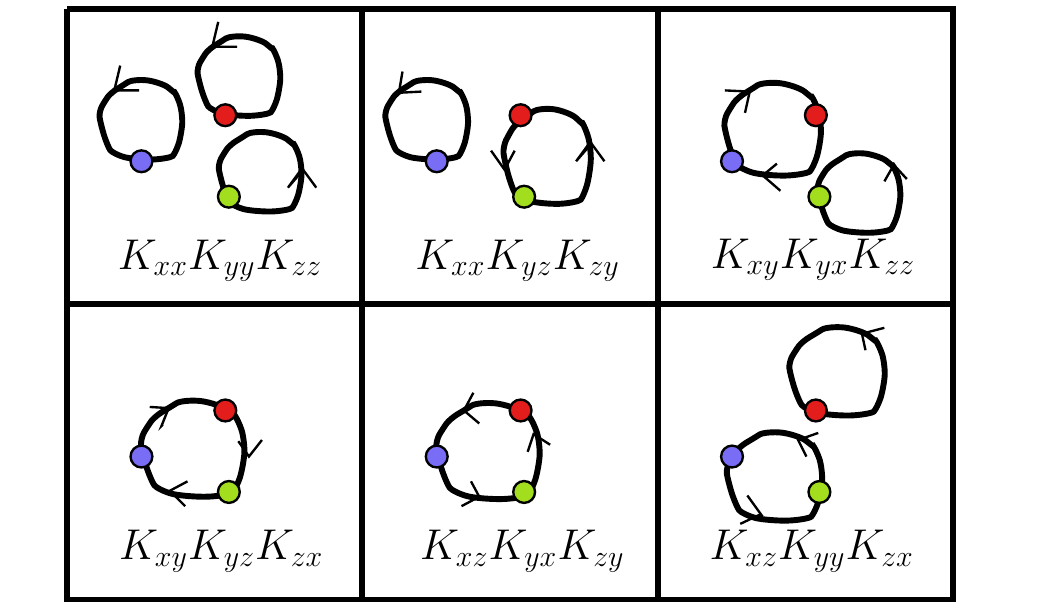}%
    \captionof{figure}{Gaussian loop soup 3-point correlations. A directed lace between two points means that the loop goes through the points in the indicated order. Point $x$ is blue, $y$ is red and $z$ is green. } 
  \label{correlations-1} 
\end{center}\end{minipage}

\begin{Remark}\rm
\label{iso1}
A Cox process is a Poisson process with a random intensity measure \cite{kingman}. Let  $\Phi_\lambda$ and $\Psi_\lambda$ be independent, centered Gaussian fields with covariance 
$\big(K_\lambda(x,y),\,x,y\in \R^d\big)$.
Using moment formulas for Gaussian random vectors it is not hard to show that a Cox process with intensity measure $\bigl(\frac12 \Phi^2_\lambda(x)+\frac12 \Psi^2_\lambda(x)\bigr)\,dx$ is permanental with $n$-point correlation function $\varphi^{\ls}_\lambda$. This is a particular case of a general identification of $\alpha$-permanental point processes as Cox processes when $2\alpha$ is a positive integer, see e.g.\/\cite{MR3791470}, Chapter 14. In particular, since correlation densities identify this point process (see the proof of Proposition \ref{teo1}), this implies that the point marginal of the Gaussian loop soup $\nu_\lambda^{\ls}$ is Cox, with (random) intensity measure $\bigl(\frac12 \Phi^2_\lambda(x)+\frac12 \Psi^2_\lambda(x) \bigr)\,dx$.
\end{Remark}

\section{Gaussian random interlacements}
\label{sri5}

In this section we construct a measure on the space of doubly infinite trajectories in $\R^d$, $d\ge 3$, analogous to the intensity measure of the Gaussian loop soup. We then define the Gaussian random interlacements as a Poisson  process with an intensity measure proportional to this measure. 
We adapt the strategy developed by Sznitman \cite{Sznitman1} to introduce the model of random interlacements in $\Z^d$. The intensity measure of the set of paths that intersect a bounded set $\Lambda$ is obtained as the projection of a measure supported on trajectories that visit $\Lambda$ at time $0$ for the first time. Precisely, the location of the first visit to $\Lambda$ is distributed according to the equilibrium measure of the associated random walk, and from this location two independent random walk trajectories are drawn, with the one running up to time $-\infty$ being conditioned to never return to $\Lambda$. We here propose an alternative, equivalent measure that chooses the location $x$ of the trajectory at time $0$ uniformly in $\Lambda$, draws two independent random walk trajectories starting at $x$, up to $-\infty$ and $+\infty$ times respectively, and assigns weight to the resulting doubly infinite path that is inversely proportional to its total number of visits to $\Lambda$. In \S\ref{pm2} we compute the density, Laplace functionals and point correlations of the point marginal of the Gaussian random interlacements.

 For $d\ge 3$, define 
\begin{align}
&\Wro:=\big\{w:\Z \to \R^d,\, \lim_{n \to \pm \infty}\|w(n)\|=\infty \big\}, \label{W}
\end{align}
the space of discrete $\R^d$-valued doubly infinite trajectories that spend finite time in compact subsets of $\R^d$. Endow  $\Wro$ with the $\sigma$-algebra ${\crW^{\rot}}$ generated by the canonical coordinates $X_n(w):=w(n)$, $n\in\Z$.
Given a bounded set $\Lambda\subset \R^d$, let $\Wro_\Lambda$ be the set of trajectories that enter~$\Lambda$,
\begin{equation}
\label{WA}
\Wro_\Lambda:=\{ w\in {\Wro}: X_n(w)\in \Lambda \text{ for some } n \in \Z\,\},
\end{equation}
and for $w \in {\Wro}$ define the entrance time of $w$ in $\Lambda$ by
\begin{align}
&T_\Lambda(w):=
                \begin{cases}
                  \inf\big\{n\in \Z,\, X_n(w) \in \Lambda \big\},&\text{if }w \in  \Wro_\Lambda,\\
                +\infty,&\text{otherwise }.
                \end{cases}
 \label{HA}%
\end{align}
Notice that $T_\Lambda(w)>-\infty$ because by definition $w$ intersects $\Lambda$ at most finitely many times.

Define the time shift $\theta: {\Wro}\to {\Wro}$ by $[\theta w](k):=w(k+1)$, and let $\theta^\ell$ be the shift by $\ell$ time units, $[\theta^{\ell}w](k):=w(k+\ell)$, $\ell  \in \Z$. Given a function $g:\Wro\to\R$ let $(\theta g)(w) := g(\theta w)$.

Given $x\in \R^d$ let $P^x$ denote the probability measure on $\big({\Wro}, \crW^{\rot}\big)$ having finite dimensional distributions 
\begin{align}
&P^{x_0}\big[X_{\ell}\in dx_\ell,\dots, X_0\in dx_{x_0},\dots, X_k\in dx_k\big]\label{psupx}\\
&\hspace{5cm}=\delta_x(dx_0)\,\prod_{i=1}^{k} p_{\frac{1}{2\alpha}}(x_{i-1},x_i)\, dx_i\,
\prod_{j=1}^{\ell} p_{\frac{1}{2\alpha}}(x_{-j+1},x_{-j})\, dx_{-j},\nn
\end{align}
where $\delta_x$ is the Dirac distribution at $x$, and $p_{\frac{1}{2\alpha}}$ is defined in \eqref{Brownian}. That is, 
$P^x$ is the law of a doubly infinite random walk in $\R^d$ that satisfies $X_0=x$ and has independent, identically distributed Gaussian increments $N(0,\frac{1}{2\alpha})$. The fact that the walk is transient implies $P^x(\Wro)=1$.   
Denote by $E^x$ the expectation with respect to $P^x$.
Since the Lebesgue measure is reversible for the random walk, we have
\begin{align}
  \label{ti1}
  \int_{\R^d} E^x [g]\,dx = \int_{\R^d} E^x [\theta g] \, dx
\end{align}
for bounded measurable test functions $g:\Wro\to\R$.

Consider  the measure $Q^{\capa}_\Lambda$ on $\Wro_\Lambda$, that integrates a test function $g: \Wro\to \R$ as
\begin{align}
\label{QA}
Q_\Lambda^{\text{cap}} g := \int_\Lambda E^x \big[g \,\one_{\{T_\Lambda=0\}} \big] dx. 
\end{align}
The function
$e_\Lambda(x):=P^x[T_\Lambda=0]$ 
is the density of a measure on $\R^d$ supported on $\Lambda$ called the \emph{equilibrium measure} associated to the Gaussian random walk. The capacity of $\Lambda$ is defined by
$\text{cap}(\Lambda) := \int_\Lambda e_\Lambda(x)\,dx$; see \cite{Sznitman1}. The equilibrium measure satisfies 
\begin{align}\label{ll1}
     e_\Lambda(y)+ \int_\Lambda e_\Lambda(x) K_1(x,y) \,dx = 1, \qquad y\in \Lambda,
\end{align}
where $K_1$ is defined in \eqref{Klambda}. To show \eqref{ll1}, consider $y\in\Lambda$ and write  
\begin{align*}
  1 &= \sum_{k \ge 0} P^y(T_\Lambda = -k)    = e_\Lambda(y) + \sum_{k \ge 1} P^y(T_\Lambda = -k) \\ 
 &= e_\Lambda(y) + \sum_{k \ge 1} \int_\Lambda P^x(T_\Lambda = 0) J^k(x,y) \,dx 
   = e_\Lambda(y)+ \int_\Lambda e_\Lambda(x) K_1(x,y) \,dx.
\end{align*}
Define the measure $Q_\Lambda^{\unif}$ on $\Wro_\Lambda$ by
\begin{align}
  \label{unif}
  Q_\Lambda^{\unif}g 
  &:= \int_\Lambda E^{x}\Bigl[\frac{g}{n_\Lambda}\Bigr]\,  dx,
\end{align}
where 
\begin{align}
\label{nAw}
  n_\Lambda(w) := \sum_{n\in\Z} {\mathbf 1}_{\Lambda}(X_n(w))
\end{align}
is the number of visits of the trajectory $w$ to $\Lambda$. The weight $ Q_\Lambda^{\unif}$ assigns to a trajectory $w$ is inversely proportional to the number of visits of $w$ to the set $\Lambda$.

  \begin{Proposition}
    \label{pro27}
    For any bounded set $\Lambda\subset\R^d$ and measurable bounded function $g:\Wro\to\R$ invariant under time shifts, $g=\theta g$, we have
    \begin{align}
      \label{pro28}
 Q_\Lambda^{\unif} g= Q_\Lambda^{\capa}g.
  \end{align}
\end{Proposition}
\begin{proof}
Write
\begin{align}
    Q^{\unif}_\Lambda g
     &= \int_\Lambda  dx \,E^{x} \left[\frac{g}{n_\Lambda}\, \hbox{$\sum_{i\le 0}$ } \one_{\{T_\Lambda=i\}}\right]\,\nn\\
    &=  \sum_{i\le 0}\int_{\R^d}  dx \,E^{x} \left[\one_\Lambda (X_0)  \one_{\{T_\Lambda=i\}}\,\frac{g}{n_\Lambda}\right] \qquad\text{by Fubini's theorem} \nn\\
    &=  \sum_{i\ge 0}\int_{\R^d}  dx \,E^{x}\left[ \one_\Lambda(X_i) \one_{\{T_\Lambda=0\}} \,\theta^{i}\Big\{\frac{g}{n_\Lambda}\Big\}\right] \qquad \text{by \eqref{ti1}}\nn\\
    &=  \int_\Lambda  dx \,E^{x} \left[\one_{\{T_\Lambda=0\}} \,\frac{g}{n_\Lambda}\, \hbox{$\sum_{i\ge 0}$ } \one_\Lambda(X_i) \right]\qquad \text{since $\theta^i\Big\{\frac{g}{n_\Lambda}\Big\}=\frac{g}{n_\Lambda}$} \nn\\
    &=  \int_\Lambda dx \,E^{x} \big[\one_{\{T_\Lambda=0\}}\,  g\big] =  Q^{\capa}_\Lambda g.\nn
    \qedhere        \end{align}
\end{proof}

\begin{Lemma}[Compatibility and additivity] Let $A\subset B$ be Borel bounded sets of $\R^d$, and let $g$ be a test function that is invariant under time shifts, $g=\theta g$. Then 
  \label{compatib}
  \begin{align}
    Q^{\capa}_{B}g\one_{\Wro_A}&=Q^{\capa}_{A}g  && \text{(compatibility)},\label{compa}\\[2mm]
    Q^{\capa}_Bg &= Q^{\capa}_Ag + Q^{\capa}_{B\setminus A}\,g \one_{(\Wro_A)^c}&&\text{(additivity)}, \label{addit}
  \end{align}
     where $(\Wro_A)^c= \Wro\setminus \Wro_A$. The same holds for $Q^{\unif}$.
  \end{Lemma}
  
 
  \begin{proof} Write $\one_{\Wro_A} = \sum_{i\in\Z}  \one_{\{T_A=i\}}$. Then    
      \begin{align*}
       Q^{\capa}_{B} g \one_{\Wro_A}
       &=\sum_{i\ge 0} \int dx \,E^{x}\big[\one_{\{T_{B}=0\}} \one_{\{T_A=i\}} \, g\big]\\
        &= \sum_{i\ge 0} \int dx \,E^{x}\big[\theta^i \big(\one_{\{T_{B}=-i\}} \one_{\{T_A=0\}} \, g\big)\big] \\
         &= \sum_{j\le 0} \int dx \,E^{x}\big[\one_{\{T_{B}=j\}} \one_{\{T_A=0\}} \, g\big] \\
    &=\int dx \,E^{x}\big[\one_{\{T_A=0\}} \, g\, \hbox{$\sum_{i\le 0}$ } \one_{\{T_{B}=i\}} \big]
    = Q^{\capa}_{A}g,  
      \end{align*}
      since
    $\one_{\{T_A=0\}}\sum_{i\le 0} \one_{\{T_{B}=i\}} = \one_{\{T_A=0\}}$.
    This proves \eqref{compa}. To get \eqref{addit} write
   \begin{align}
      Q^{\capa}_{B}g &= Q^{\capa}_{B}g(\one_{\Wro_A} + \one_{(\Wro_{A})^c}) 
      = Q^{\capa}_A g+ Q^{\capa}_{B}g\one_{(\Wro_{A})^c} = Q^{\capa}_Ag + Q^{\capa}_{B\setminus A}g\one_{(\Wro_{A})^c}. \nn
    \end{align}
    Since $g\one_{\Wro_A}=\theta(g\one_{\Wro_A})$, $g=\theta(g)$ and $g\one_{(\Wro_{A})^c}=\theta(g\one_{(\Wro_{A})^c})$, Proposition \ref{pro27} implies that \eqref{compa} and \eqref{addit} hold for $Q^{\unif}$ as well.
  \end{proof}
Now let us identify trajectories that differ by time shift: given two doubly-infinite trajectories $w,\,w' \in {\Wro}$, we say that $w\sim w'$ if there exists $k\in \Z$ such that $w'=\theta^kw$. Let 
\begin{equation}
\label{doubly inf}
\tW:={\Wro}/\sim
\end{equation}
be the space of trajectories modulo time shift, $\pi:{\Wro}\to \tW$ the projection, and $\crW$ the push-forward $\sigma$-algebra on $\tW$. Given $\Lambda\subset \R^d$ let 
${\tW}_\Lambda=\pi(\Wro_\Lambda)$. If $g:\Wro\to \R$ is shift invariant, then it can be extended to $\tilde{g}: \tW\to \R$ by $\tilde{g}(\tilde{w})=g(w)$, for any choice of representative $ w \in \pi^{-1}(\tilde w)$.
\begin{Proposition}[Infinite volume measure]
\label{THnu}
There exists a unique $\sigma$-finite measure $Q^{\ri}$ on $(\tW, \crW)$ such that for each bounded set $\Lambda\subset \R^d$
\begin{align}
\label{nu}
{\mathbf 1}_{\tW_\Lambda} Q^{\ri}=\pi_* Q^{\capa}_\Lambda=\pi_* Q^{\unif}_\Lambda,
\end{align}
where $\pi_*Q_\Lambda^{\capa}$ and $\pi_*Q_\Lambda^{\unif}$ denote the push-forward measures.
\end{Proposition}
\begin{proof}
Let $\tilde{g}:\tW\to \R$ and define $g:\Wro\to \R$ by $g=\tilde{g}\circ \pi$. Then $\theta g=g$ and 
\begin{align*}
\pi_* Q_\Lambda^{\capa} \tilde{g}&=Q_\Lambda^{\capa} \tilde{g}\circ \pi=Q_\Lambda^{\capa} g=Q_\Lambda^{\unif} g=Q_\Lambda^{\unif}\tilde{g}\circ \pi   =\pi_* Q_\Lambda^{\unif} \tilde{g}
\end{align*} 
by Proposition \ref{pro27}. This proves the second equality in \eqref{nu}.
 
Let $\{\Lambda_n\}_{n\ge 1}$ be an increasing sequence of bounded sets such that $\Lambda_n \nearrow_{n \to \infty} \R^d$. Then $\tW=\bigcup_{n\ge 1} \tW_{\Lambda_n}$ and uniqueness of the measure satisfying \eqref{nu} follows. Define $Q^{\ri}$ on $\tW_{\Lambda_n}$ by
\begin{align}
\label{onAn}
{\mathbf 1}_{\tW_{\Lambda_n}} Q^{\ri}:=\pi_* Q^{\capa}_{\Lambda_n}.
\end{align}
Let $\Lambda$ be a bounded set and take $n$ sufficiently large such that $\Lambda_n\supset \Lambda$. Then
\begin{equation}
\label{ntom}
\one_{\tW_{\Lambda}}\one_{\tW_{\Lambda_n}}\big(\pi_* Q^{\capa}_{\Lambda_n}\big)=\one_{\tW_{\Lambda}}\big(\pi_* Q^{\capa}_{\Lambda_n}\big)=\pi_* \one_{\tW_{\Lambda}}Q^{\capa}_{\Lambda_n} =\pi_*Q^{\capa}_{\Lambda},
\end{equation}
where the last identity follows from \eqref{compa}. In the case when $\Lambda=\Lambda_m$ for some $m<n$, \eqref{ntom} proves that the definition \eqref{onAn} is consistent and that the measure $Q^{\ri}$  defined in \eqref{onAn} satisfies \eqref{nu}. By \eqref{nu}, $Q^{\ri}(\tW_{\Lambda_n})=Q^{\capa}_{\Lambda_n}({\tW})=\text{cap}(\Lambda_n)<\infty$, which proves the ${\sigma}$-finite property.
\end{proof}
For a trajectory ${\tilde w}\in\WW$, let $\{\tw\}:= \{w(n):n\in\Z\}$ for any  $w\in \pi^{-1}({\tilde w})$, this is just the set of points in the trajectory $\tw$. Recall $n_\Lambda(\tw)$ is the number of points in $\{\tw\}\cap\Lambda$. Denote by $\kW$ the space of locally finite interlacement configurations,
\begin{align}
  \kW :=\{\Gamma\subset \WW: \tsum_{\tw\in\Gamma}\,n_\Lambda(\tw)\,<\,\infty, \hbox{ for all compact }\Lambda\subset\R^d\}. 
\end{align}
\begin{Definition}[Gaussian random interlacements]\rm
\label{randomi}
Let $d\ge 3$ and $\rho>0$. The \emph{Gaussian random interlacement} process at level $\rho$ is the Poisson point process on the space $\kW$
, with intensity measure $\rho Q^{\ri}(d\tilde{w})$. We will use $\Gamma^{\ri}_\rho$ to denote a configuration of the random interlacements at level $\rho$, and $\mu^{\ri}_\rho$ to denote  its law.
\end{Definition}

This definition is slightly simpler than the one proposed by Sznitman \cite{Sznitman1} in $\Z^d$. The original definition considers a point process in the product space $\R_{\ge 0} \times \tW$ with intensity measure $du \otimes Q^{\ri}(d\tilde{w})$, so that each trajectory $\tilde{w}$ is sampled with an associated level $u>0$ used to couple realizations of the model at different densities. For our purposes it is enough  to a priori fix the level $\rho$.

\subsection{Point marginal of the Gaussian random interlacements} 
\label{pm2}
We now study the point marginal of the random interlacements.
Define the random point process
\begin{align}
  \label{chi}
  \sX^{\ri}_\rho := \bigcup_{\tw\in\Gamma^{\ri}_\rho}\{\tw\},\quad \text{ and let }\quad \nu^{\ri}_\rho \quad \text{be its law}.
\end{align}
Recall that $n_\Lambda(\sX)$ denotes the number of points of a locally finite set of points $\sX\subset \R^d$ in $\Lambda$. The mean density of 
$\sX^{\ri}_\rho$ in $\Lambda$ is defined by 
\begin{equation}
\label{density ri}
\varrho^{\ri}_\rho(\Lambda)\;:=\; \frac{1}{|\Lambda|} \nu^{\ri}_\rho \,n_\Lambda \;=\;
\frac{1}{|\Lambda|}\int \sum_{\tilde{w}\in \Gamma} n_\Lambda(\tilde{w})\,\mu^{\ri}_\rho(d\Gamma).
\end{equation}

\begin{Lemma} [Point density]
\label{density RI}
For any bounded Borel set $\Lambda \subset \R^d$ we have
 \begin{align}
      Q^{\ri}(n_\Lambda) &= |\Lambda|,  \label{cuf3}\\[1mm]
    \varrho^{\ri}_\rho(\Lambda)& =\rho.\label{density RI2}
    \end{align}
In particular the mean density does not depend on the set $\Lambda$.
\end{Lemma}
\begin{proof}
To show \eqref{cuf3}, write
    \begin{align}
      Q^{\ri}(n_\Lambda) = Q^{\unif}_\Lambda (n_\Lambda) = \int_\Lambda E^x\left[\frac{n_\Lambda}{n_\Lambda}\right]dx= |\Lambda|.
    \end{align}
To show \eqref{density RI2},   we compute the mean number of points of the process in $\Lambda$ as
\begin{align*}
   \int \sum_{\tilde{w}\in \Gamma}  n_\Lambda(\tilde{w})\,\mu^{\ri}_\rho(d\Gamma)
  =\rho Q^{\ri}(n_\Lambda)
=\rho |\Lambda|,
\end{align*}
where the first identity follows from Campbell's theorem and the second one by \eqref{cuf3}.
\end{proof}

We next compute the Laplace functionals of $\sX_\rho^{\ri}$, denoted by
\begin{align}
  \cL_{\rho}^{\ri}(\phi)&:=\int \nu_{\rho}^{\ri}(d\sX) \exp\Big\{ -\sum_{x \in \sX} \phi(x) \Big\}. \nn
\end{align}
To simplify notation, in the rest of this section we will write $K$ instead of $K_1$, the kernel defined in \eqref{Klambda}.
\begin{Proposition}[Laplace functionals]
\label{riLaplace}
Let $\phi:\R^d \to \R_{\ge 0}$ be a measurable function with compact support and satisfying $ K(0,0) \|e^{-\phi}-1\|_1<1$. Then 
\begin{align}
\label{laplace1}
\cL_{\rho}^{\ri}(\phi)
&=\exp\Big\{ \rho\sum_{n\ge 1} \ \int_{(\R^d)^n} \prod_{k=1}^{n-1} K(x_k,x_{k+1})\,\prod_{k=1}^n\big(e^{-\phi(x_k)}-1\big)\, dx_1 \dots dx_n\Big\}, 
\end{align}
with $ \prod_{k=1}^{0} K(x_k,x_{k+1}):=1$.
\end{Proposition}
\begin{proof}
  Define $\Phi: \tW\to \R$  by 
\begin{equation}
\label{capital F}
\Phi(\tilde{w})=\sum_{n\in \Z} \phi(w(n)),\quad \text{for any }w\in\pi^{-1}(\tw).
\end{equation} 
The definition does not depend on the choice of representative $w\in \pi^{-1}(\tilde{w})$. Then,
\begin{align}
\label{laplace2}
\cL^{\ri}_\rho(\phi)=\int_{\tW} \exp\Bigl\{ -\sum_{\tilde{w}\in \Gamma} \Phi(\tilde{w})\Bigr\} \, \mu_{\rho}^{\ri}(d\Gamma)=
\exp\Big\{ \rho \int_{\tW} \big(e^{-\Phi(\tilde{w})}-1\big)Q^{\ri}(d\tilde{w})\Big\},
\end{align}
by Campbell's theorem.
Let $\Lambda=\text{supp}(\phi)\subset \R^d$ which is compact by hypothesis. Then
\begin{align}
\label{long}
\int_{\tW} \big(e^{-\Phi(\tilde{w})}-1\big)Q^{\ri}(d\tilde{w})&= Q^{\capa}_\Lambda\Big( \prod_{i=0}^\infty e^{-\phi(w(i))}\,-1\Big),
\end{align}
 by Proposition \ref{THnu}. For $w \in \Wro$, let us write
\begin{align}
\label{tired}
 \prod_{i=0}^\infty e^{-\phi(w(i))}\,-1&= \prod_{i=0}^\infty \Big[\big(e^{-\phi(w(i))}-1\big)+1\Big]\,-1 \notag\\ 
 &=\sum_{ \emptyset\neq S\subset \N_0}\  \prod_{i\in S} \big(e^{-\phi(w(i))}-1\big) \notag \\
&=\sum_{n\ge 1} \ \sum_{0\le i_1<i_2<\dots\,<i_n}\   \prod_{j=1}^n \big(e^{-\phi(w(i_j))}-1\big).
\end{align}
The trajectory $w \in \Wro$ enters the support of $\phi$ finitely many times and the factor $e^{-\phi(x)}-1$ vanishes outside this set, hence the sum on the second line above may be  restricted to finite sets $S\subset \N_0$ to get the third line. 

Let $\Lambda^k=\Lambda\times \dots \times \Lambda$ denote the product of $k$ copies of $\Lambda$, $k\ge 1$. We replace the last expression from \eqref{tired} in \eqref{long}, and exchange the order of the sums and the integral, which is justified later,  to obtain
\begin{align}
&\frac{1}{\rho}\log\cL^{\ri}_\rho(\phi) = \nn\\
&\ \ =\sum_{n\ge 1} \ \sum_{0\le i_1<i_2<\dots\,<i_n}\int_{\Wro}  \prod_{j=1}^n \big(e^{-\phi(w(i_j))}-1\big) Q^{\capa}_\Lambda(dw) \nn\\
&\ \ =\sum_{n\ge 1}\sum_{\substack{1\le \ell_1\\ \dots \\ 1\le \ell_{n-1}}} \int_{\Lambda^n} e_\Lambda(x_1) \prod_{k=1}^{n-1} J_1^{\ell_k}(x_k, x_{k+1})
\prod_{k=1}^n\big(e^{-\phi(x_k)}-1\big)\, dx_1 \dots dx_n\label{ll22}\\
&\ \ \quad+\sum_{n\ge 1}\sum_{\ell\ge 1}\sum_{\substack{1\le \ell_1\\ \dots \\ 1\le \ell_{n-1}}} \int_{\Lambda^{n+1}} e_\Lambda(x_0) J_1^\ell(x_0,x_1) \prod_{k=1}^{n-1} J_1^{\ell_k}(x_k, x_{k+1})
\prod_{k=1}^n\big(e^{-\phi(x_k)}-1\big)\, dx_0 \dots dx_n  \nn\\
&\ \ =\sum_{n\ge 1}  \sum_{\substack{1\le \ell_1\\ \dots \\ 1\le \ell_{n-1}}} \  \int_{\Lambda^{n}} \prod_{k=1}^{n-1} J_1^{\ell_k}(x_k, x_{k+1})
\prod_{k=1}^n\big(e^{-\phi(x_k)}-1\big)\, dx_1 \dots dx_n\qquad \text{by \eqref{ll1} }\nn\\
&\ \ =\sum_{n\ge 1}\  \int_{\Lambda^{ n}} \prod_{k=1}^{n-1} K(x_k,x_{k+1})\,\prod_{k=1}^n\big(e^{-\phi(x_k)}-1\big)\, dx_1 \dots dx_n\label{j87}\\
&\ \ =\sum_{n\ge 1} \ \int_{(\R^d)^{n}} \prod_{k=1}^{n-1} K(x_k,x_{k+1})\,\prod_{k=1}^n\big(e^{-\phi(x_k)}-1\big)\, dx_1 \dots dx_n,\nn
\end{align}
where the first term in expression \eqref{ll22} corresponds to $i_1=0$ and in the second term there is a change of variables $l_j = i_j-i_{j-1}$. If we replace $e^{-\phi(x_k)}-1$ in the last expression by $|e^{-\phi(x_k)}-1|$, we obtain
a convergent series thanks to the condition $K(0,0) \| e^{-\phi} - 1 \|_1 < 1$. Thus we can 
apply the dominated convergence theorem to exchange the integral with infinite sums. The result follows by replacing the last expression in~\eqref{laplace2}.
\end{proof}

\subsubsection{Point correlations}
\label{pc-gi}

\begin{Proposition}[Point correlations]
  \label{pcri}
  The $n$-point correlation density of the point process $\sX_{\rho}^{\ri}$ is given by 
\begin{align}
  \label{18}
  \varphi_{\rho}^{\ri}(x_1,\dots,x_n)\; =\; \sum_{P\in \cP_n} \; \prod_{I=\{i_1,\dots,i_m\} \in P} \; \sum_{\sigma \in \cS_m} \;\rho\, K(x_{i_{\sigma(1)}},x_{i_{\sigma(2)}})\dots K(x_{i_{\sigma(m-1)}},x_{i_{\sigma(m)}}),
\end{align}
where we recall that $\cP_n$ is the set of partitions of $\{1,\dots,n\}$ given in \eqref{parti} and $\cS_m$ is the set of permutations of $\{1,\dots,m\}$.
\end{Proposition}
This result is a consequence of Proposition \ref{riLaplace}. A direct proof of \eqref{18} can be obtained as in the loop soup case, Proposition \ref{ls-correl}. For instance, for $n=3$, taking $x,y,z$ distinct points in $\R^d$ and denoting $K_{xy}=K(x,y)$ we have 
\begin{align}
   \varphi_{\rho}^{\ri}(x,y,z)
                  &= \rho^3 +2 \rho^2 K_{xy} + 2 \rho^2 K_{yz}+ 2 \rho^2 K_{xz}+ 2\rho K_{xy}K_{yz}+ 2 \rho K_{xz}K_{zy}+ 2\rho K_{zx}K_{xy}. \label{ri19}
\end{align}
The terms in \eqref{ri19} are associated to the partition of the set of trajectories containing $x,y,z$, illustrated in Fig.~\ref{correlations-ri}. For instance, the first term $\rho^3$ is the density of the event  $\{\Gamma: x,\,y$ and $z$ belong to different trajectories in~$\Gamma\}$, while $2\rho K_{xy}K_{yz}$ is the density of $\{\Gamma: x,\,y$ and $z$ are visited by one trajectory in $\Gamma$ in this or the reverse order$\}$, the factor 2 accounts for the two possible orders.

     \noindent\begin{minipage}{\linewidth}\begin{center}
         \includegraphics[scale=.5]{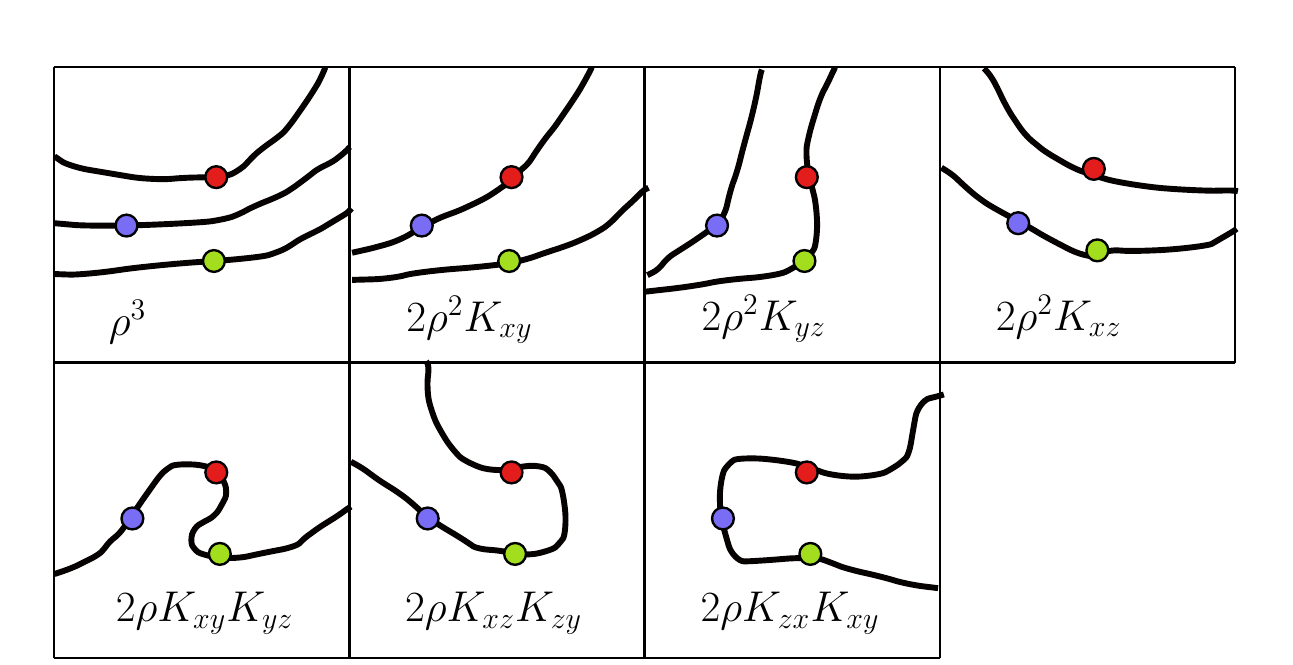}%
\captionof{figure}{Gaussian random interlacements 3-point correlations. Each square represents a part of the set of trajectories containing the points $x$, $y$, $z$. The point $x$ is blue, $y$ is red and $z$ is green.
\label{correlations-ri}}
\end{center}\end{minipage}

\section{Infinite volume Gaussian random permutation}
\label{infvol}

In this section we define the infinite volume Gaussian random permutation. In  \S\ref{grp} we show that the Gaussian random permutation has a Markov property and in \S\ref{gm1} that it is a Gibbs measure for the specification induced by \eqref{finite-box}. Let $\kG$ be defined by
\begin{align}
  \kG:= \{\Gamma\cup \Gamma': \Gamma\in\kD,\;\Gamma'\in \kW\},
\end{align}
this is the space of superpositions of locally finite loop and interlacement configurations. 
A configuration $\Gamma\in\kG$ is in bijection with the spatial permutation $(\sX,\sigma)$ defined by
\begin{align}
  \label{86b}
  \sX:= \cup_{\gamma\in \Gamma} \{\gamma\};\quad \sigma(x) := \gamma(x), \quad \hbox{if $x\in\{\gamma\}$ and $\gamma\in\Gamma$}.
\end{align}
We will abuse notation and use indistinctly $\Gamma$ and $(\sX,\sigma)$ to refer to the spatial permutation.

Recall the definition of the critical density $\rho_c = \left(\frac{\alpha}{\pi}\right)^\frac{d}{2} \sum_{k \geq 1} k^{-\frac{d}{2}}$ in \eqref{rcrit}, so that 
$\rho_c=\infty$ if $d=1,\,2$. Define
\begin{align}
  \label{sol-lambda}
  \lambda(\rho_c)&:=1, \notag \\
  \lambda(\rho) &:= \hbox{the solution to }\rho = \left(\frac{\alpha}{\pi}\right)^{\frac{d}{2}}
  \sum_{k \geq 1} \lambda^k\,k^{-\frac{d}{2}},\quad \rho\in(0,\rho_c).
\end{align}
Given two independent processes $\Gamma$ and $\Gamma'$ with laws $\mu$ and $\mu'$, we denote by
\begin{align}
\label{conv}
\mu * \mu' \; := \; \text{law of }\; \Gamma\cup\Gamma',
\end{align}
the law of the superposition of the processes.

\begin{Definition}
  \label{def}\rm
The \emph{Gaussian random permutation} in $\R^d$ at density
$\rho>0$ is defined by
\begin{align}
  \label{grp7}
  \mu_\rho :=
  \begin{cases}
    \mu^{\ls}_{\lambda(\rho)},
    &\hbox{if } d\le2 \text{ and } \rho>0,\;\text{ or }\; d\ge3\;\text{and } \rho\le \rho_c, \\[2mm]
    \mu^{\ls}_1 * \mu^{\ri}_{\rho-\rho_c},
    &\hbox{if }\;d\ge3\;\text{and } \rho> \rho_c,
  \end{cases}
\end{align}
where the Gaussian loop soup $\mu^{\ls}_\lambda$ and the Gaussian random interlacements $\mu^{\ri}_{\rho-\rho_c}$ are as in Definitions \ref{loop soup} and \ref{randomi}, respectively. 
\end{Definition}

\subsection{Markov property}
\label{grp}
We show that for all $\rho>0$ the Gaussian spatial permutation $\mu_\rho$ is Markov in the sense 
of Proposition \ref{markov}.

We define the ``inside'' and  ``outside'' components of a spatial permutation $\Gamma=(\sY,\kappa)\in\kG$ with respect to a bounded Borel set $\Lambda\subset\R^d$.  Define the sets of points 
\begin{align}
   \arraycolsep=2pt\def\arraystretch{.2}
 \begin{array}{rlcl}
   I\sY&:=\sY\cap\Lambda,&\ \ &\hbox{points in $\Lambda$; red,}\\[2mm]
     O\sY&:=\sY\cap\Lambda^c,& &\hbox{points in $\Lambda^c$; purple and yellow,}\\[2mm]
  U\sY&:=\{u\in\sY\cap\Lambda^c: \kappa(u)\in\Lambda\},& &\hbox{points in $\Lambda^c$ with image in $\Lambda$; yellow $u$,}\\[2mm]
  V\sY&:=\{v\in\sY\cap\Lambda^c: \kappa^{-1}(v)\in\Lambda\},& &\hbox{points in $\Lambda^c$ with pre-image in $\Lambda$; yellow $v$,}
   \end{array}\label{IX}
\end{align}
and the maps
\begin{align}
\label{cmaps}
    \arraycolsep=1.5pt\def\arraystretch{.2}
 \begin{array}{rlcrlll}
     I\kappa:& I\sY\cup U\sY \to I\sY\cup V\sY,&\ \
   & I\kappa(x) &= \kappa(x); &\ \
   &\hbox{red arrows},
   \\[2mm]
   O\kappa:& O\sY\setminus U\sY\to O\sY\setminus V\sY,&\ \
     & O\kappa(x) &= \kappa(x);  &\ \
   &\hbox{purple arrows}.
   \end{array}
\end{align}
Colors and labels $u,v$ refer to Fig.~\ref{gibbs-3}. 
We have dropped from the notation the dependence on $\Lambda$, the dependence on $\kappa$ of the definition of $U\sY$ and $V\sY$ and the dependence on $\sY$ of the definition of $O\kappa$ and $I\kappa$. 
Denote the inside and outside components of $\Gamma=(\sY,\kappa)$ with respect to $\Lambda$, and the boundary points, by
\begin{align}
\label{oo36}
I(\sY,\kappa) := (I\sY, I\kappa), \quad  O(\sY,\kappa) := (O\sY, O\kappa), \quad (U,V)(\sY,\kappa):=(U\sY,V\sY) .
\end{align}
   
   \noindent\begin{minipage}{\linewidth}
     \begin{center}
       \includegraphics[scale=.5]{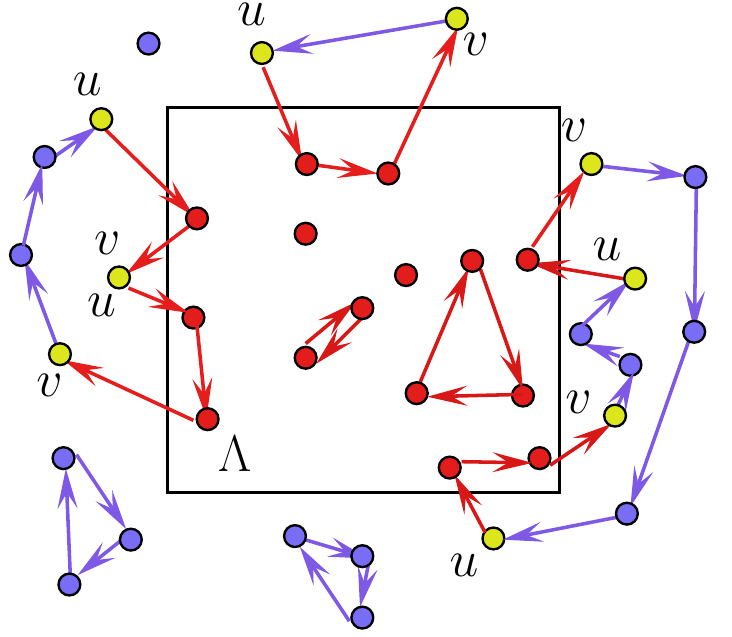}\quad\quad\quad%
         \includegraphics[scale=.5]{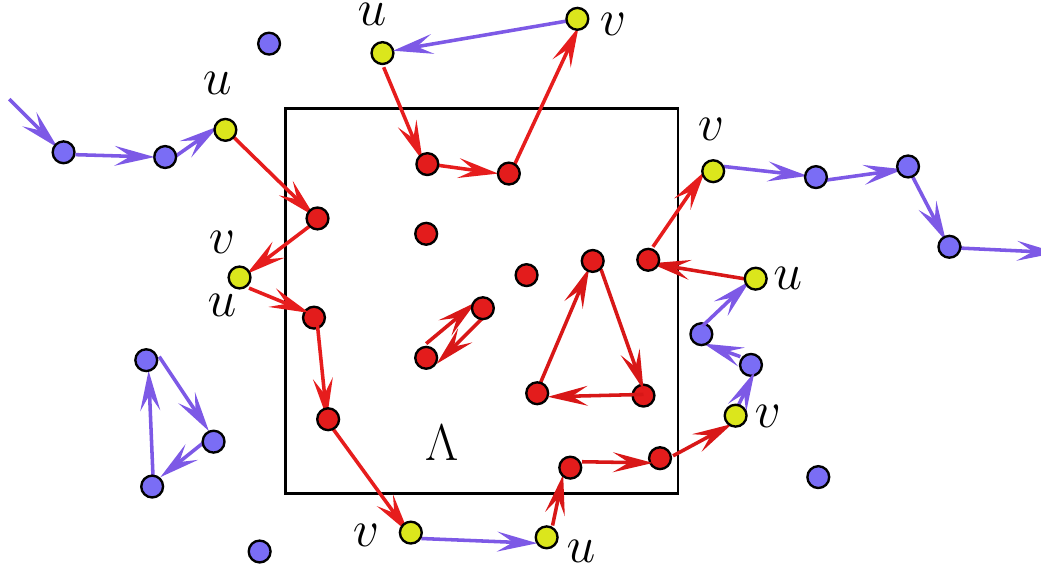}
\\
\captionof{figure}{
Left: decomposition of boundary loops of a loop soup with respect to the set $\Lambda$. Right: decomposition of boundary loops and infinite paths intersecting a bounded set $\Lambda$. 
  }
\label{gibbs-3}
\end{center}
\end{minipage}

In the next proposition we show that the law of the inside component of a random spatial permutation given its outside component is determined by the outside boundary points. 
Let $\Lambda$ be a bounded Borel set and call $\crF_{\tti \Lambda}$, $\crF_{\ttb \Lambda}$ and $\crF_{\tto \Lambda}$ the $\sigma$-algebras generated by the maps $\Gamma\mapsto I\Gamma$, $\Gamma\mapsto (U,V)\Gamma$ and $\Gamma\mapsto O\Gamma$, respectively.
\begin{Proposition}[Markov property]
  \label{markov}
  The Gaussian random permutation $\mu_\rho$ satisfies the following Markovian property: let $\Lambda$ be a bounded Borel set and $g$ be an $\crF_{\tti \Lambda}$ measurable bounded function. Then
  \begin{align}
    \label{markov53}
   \mu_\rho (g|\crF_{\tto \Lambda}) =  \mu_\rho (g|\crF_{\ttb \Lambda}).
 \end{align}

 \end{Proposition}
 Informally, the statement is: if we condition on the purple and yellow labeled points and the purple arrows of Fig.~\ref{gibbs-3}, then the law of the red points and red arrows is determined by the positions and labels of the yellow points.

 \begin{proof} 

   Recall that $\DD_\Lambda$ consists of the loops contained in $\Lambda$, and $\tW_\Lambda$ is the set of trajectories intersecting $\Lambda$. Call $\BB_\Lambda=\DD\setminus (\DD_\Lambda\cup \DD_{\Lambda^c})$ the set of loops that intersect both $\Lambda$ and $\Lambda^c$, the boundary loops. We have the following partition:
   \begin{align}
     \label{part54}
     \DD\cup \tW = \DD_\Lambda \,\dot\cup\, \DD_{\Lambda^c} \,\dot\cup\, \BB_\Lambda \,\dot\cup\, \tW_\Lambda\,\dot\cup\, (\tW\setminus \tW_{\Lambda}).
   \end{align}
Let $\Gamma$ be a sample of $\mu_\rho$. The intersection of $\Gamma$ with the above sets produces five independent Poisson processes. The loops contained in $\Lambda$ 
are independent of $\crF_{\tto \Lambda}$; on the other hand, the loops contained in $\Lambda^c$ and the infinite trajectories not intersecting $\Lambda$ are independent of $\crF_{\tti \Lambda}$. The problem thus reduces to proving identity \eqref{markov53} for 
 $g$ measurable with respect to ${\crF}_{\tti \Lambda}\big|_{\BB_{\Lambda}\cup \tW_\Lambda}$, the $\sigma$-algebra generated by the mapping $\Gamma\to I\Gamma$ restricted to $\BB_{\Lambda}\cup \tW_\Lambda$.

Define the inside and outside components of $\Gamma$ by 
\begin{align}
  \tI\Gamma&:=\big\{\eta=(x_0,\dots, x_{\ell(\eta)+1}): x_0\in U\Gamma;\,  x_i=\kappa^i(x_0),\,i=1,\dots,\ell(\eta)+1\}, \label{dd90}
             \\
  \tO\Gamma&:=\big\{\eta'=(y_0,\dots, y_{\ell'(\eta')+1}): y_0\in V\Gamma;\,  y_i=\kappa^i(y_0),\,i=1,\dots,\ell'(\eta')+1\} \label{dd91}
  \\
  &\qquad\qquad \cup\,\big\{\eta''= (\dots,z_{-1},z_0): z_0 \in U\Gamma, \text{ and }  z_{-i}=\kappa^{-i}(z_0) \in \Lambda^c \text{ for all }i\ge1\bigr\}\label{dd92}
 \\
  &\qquad\qquad \cup\,\big\{\eta''= (z_0,z_1,\dots): z_0 \in V\Gamma, \text{ and }  z_i=\kappa^{i}(z_0) \in \Lambda^c \text{ for all }i\ge1\bigr\},\label{dd93}
\end{align}
where $\ell(\eta):=\min\{\ell\ge 1: \kappa^{\ell+1}(u)\in V\Gamma\}$ and $\ell'(\eta'):=\min\{\ell\ge 0: \kappa^{\ell+1}(v)\in U\Gamma\}$. In particular $x_{\ell(\eta)+1}\in V\Gamma$ and $y_{\ell'(\eta')+1}\in U\Gamma$.
In Fig.~\ref{gibbs-3} each element of $\tI\Gamma$ is a red path from a $u$ point to a $v$ point, while the elements of $\tO\Gamma$ are purple paths joining a $v$ point to a $u$ point and semi-infinite purple paths contained in $\Lambda^c$ starting at a $v$ point or ending at a $u$ point.
Each path in the inside boundary $\tI \Gamma$ contains at least one point in $\Lambda$, so that $\ell(\eta)\ge1$. The outside boundary $\tO\Gamma$ may contain a path  $(v,u)$ with $v\in V\Gamma$, $u\in U\Gamma$ so that $\ell(\eta')=0$ in this case (see the upper $u,v$ in Fig.~\ref{gibbs-3}). There is no path when  $u=v\in U\Gamma\cap V\Gamma$. Notice that for $\Gamma \subset \BB_{\Lambda}\cup \tW_\Lambda$, $\tO\Gamma=O\Gamma$, as $\Gamma$ does not have any loops or trajectories that do not intersect $\Lambda$. On the other hand, 
$\tI\Gamma$ differs from $I\Gamma$ in that the former includes the $(U,V)$ boundary points of the loops and trajectories in $\Gamma$.

 We first prove the result for $\rho\le\rho_c$, that is the loop soup with fugacity $\lambda=\lambda(\rho)<1$  and $\lambda=1$ for $d\ge3$. 
Given $\ell\ge0$, denote
\begin{align}
  p(z_0,\dots,z_{\ell+1}) := \tprod_{i=0}^\ell \,p_{1/2\alpha}(z_i,z_{i+1}), 
                                \label{pes3}
 \end{align}
  where $p_{1/2\alpha}$ is the Gaussian density in \eqref{Brownian}. 
Recall the loop soup weight $f _{\lambda}^{\ls}$ defined in  \eqref{fls1} and observe that for $\Gamma\subset \BB_\Lambda$, we have
  \begin{align}
   f_{\lambda}^{\ls} (\Gamma)&=   f_{\lambda}^{\ls,\tti} (\tI\Gamma)\,  f_{\lambda}^{\ls,\tto} (\tO\Gamma),\label{fact1}\\[3mm]
    f_{\lambda}^{\ls,\tti} (\tI\Gamma):=  \prod_{\eta\in\tI\Gamma} \lambda^{\ell(\eta)}\,p(\eta),\quad&\quad
    f_{\lambda}^{\ls,\tto} (\tO\Gamma):=  \lambda^{n(U\cup V)}\prod_{\eta'\in\tO\Gamma} \lambda^{\ell(\eta')}\,p(\eta').\label{fact2}
  \end{align}
  where $n(U\cup V)$ denotes the cardinality of the set $U\cup V$.
By \eqref{dGamma}, if $\fF\subset \DD$ has finite intensity measure $Q^{\ls}_\lambda(\fF)<\infty$, when $g(\Gamma)$ is determined by the loops of  $\Gamma$ that belong to $\fF$, we have
  \begin{align}
    \mu^{\ls}_\lambda g &= \int_{\kD_\fF}\, g(\Gamma)\, f(\Gamma)\, d\Gamma,\label{lsls1}\\
  f(\{\gamma_1,\dots,\gamma_{\ell}\})&:= e^{-Q_\lambda(\fF)}\,\one_{\{\gamma_1<\dots<\gamma_\ell\}}\, f^{\ls}_\lambda(\{\gamma_1,\dots, \gamma_\ell\}).\label{lsls2}
  \end{align}

Let now $g$ be ${\crF}_{\tti \Lambda}\big|_{\BB_{\Lambda}}$ measurable, and consider a bounded, ${\crF}_{\tto \Lambda}\big|_{\BB_{\Lambda}}$-measurable function $h$. With the lexicographic order in $\R^d$, by \eqref{lsls1}, we have
\begin{align}
\mu^{\ls}_{\lambda}(gh)&=e^{-Q_\lambda^{\ls}(\DD_\Lambda)}\sum_{k\ge 0} \int_{(\Lambda^c)^k\times(\Lambda^c)^k} \one_{\{u_1\le \dots\le u_k\}} \,du_1\dots du_k \,\, dv_1\dots dv_k \notag\\
&\hspace{.8cm} \times \int_{ \tI(\kD_\Lambda)(u_1,\dots,u_k;\,v_1,\dots,v_k)} g(\Gamma_{\tti})f_\lambda^{\ls,\tti}(\Gamma_{\tti})\,d\Gamma_{\tti} 
\times \int_{ \tO(\kD_\Lambda)(u_1,\dots,u_k;\,v_1,\dots,v_k)}h(\Gamma_{\tto}) f_\lambda^{\ls, \tto}(\Gamma_{\tto})\, d\Gamma_{\tto} \label{S51}
\end{align}
where 
$\tI(\kD_\Lambda)(u_1,\dots,u_k;\,v_1,\dots,v_k)$ denotes the set of inside components of permutations $\Gamma \in \kD_\Lambda$ with boundary conditions $(u_1,v_1),\dots (u_k, v_k)$, and similarly $\tO(\kD_\Lambda)(u_1,\dots,u_k;\,v_1,\dots,v_k)$ is the set of outside components of permutations compatible with the boundary conditions $(u_1,v_1),\dots (u_k, v_k)$, $\kD_\Lambda$ as in \eqref{dl28}. For a bounded test function $\tg$, the integral in \eqref{S51} is defined as
\begin{align}
\label{S52}
  \int_{ \tI(\kD_\Lambda)(u_1,\dots,u_k;\,v_1,\dots,v_k)} \tg(\Gamma_{\tti})\,d\Gamma_{\tti}
  =\int \tg(\{\eta_1,\dots,\eta_k\}) \prod_{j=1}^k\one_{\{(U,V)(\eta_j)=(u_j,v_j)\}}\,d\eta_1\dots d\eta_k,
\end{align}
where the integral over inside paths from $u$ to $v$ is given by
\begin{align}
\label{S53}
  \int_{\tI(\kD_\Lambda)(u;v)} \bg(\eta)\,d\eta
  =\sum_{\ell\ge 1} \int_{\Lambda^\ell}  \bg\big((u,x_1,\dots,x_\ell,v)\big)\, dx_1 \dots dx_\ell,
\end{align}
$ \bg$ a bounded test function.
Similarly,
\begin{align}
\label{S54}
  \int_{ \tO(\kD_\Lambda)(u_1,\dots,u_k;\,v_1,\dots,v_k)} \thh(\Gamma_{\tto})\,d\Gamma_{\tto}
  =\int \thh(\{\eta'_1,\dots,\eta'_k\}) \prod_{j=1}^k\one_{\{(V,U)(\eta'_j)=(v_j,u_j)\}}\,d\eta'_1\dots d\eta'_k,
\end{align}
where the integration over outside paths from $v$ to $u$ is defined by
\begin{align}
\label{S55}
  \int_{\tO(\kD_\Lambda)(u;v)} \bh(\eta)\,d\eta
  =\sum_{\ell'\ge 0} \int_{(\Lambda^c)^{\ell'}} \bh\big((v,y_1,\dots,y_{\ell'},u)\big)\, dy_1 \dots dy_{\ell'} \,.
 \end{align}
From \eqref{S51}, with definitions \eqref{S52}-\eqref{S55}, it is clear that 
\begin{align}
\label{S56}
  \mu_\lambda^{\ls}(gh)
  =E\bigl[ E\bigl[g(\tI\Gamma)|{\crF}_{\ttb \Lambda}\bigr]h\bigr],
\end{align}
and $\tI\Gamma$, $\tO\Gamma$ are independent given $(U,V)(\Gamma)$, the Markov property.
  
We turn now to the supercritical case, that is,  $d\ge 3$ and $\rho>\rho_c$. The Gaussian random permutation $\mu_\rho$ is the law of the superposition of independent realizations of the Gaussian loop soup at density $\rho_c$ (fugacity $\lambda=1$) and  
the Gaussian random interlacements at density $\rho-\rho_c$. 

A trajectory $w\in \tW_\Lambda$ enters $\Lambda$ for the first time from a point $s(w)$ and visits $\Lambda$ for the last time before jumping to a point $t(w)$. More precisely, 
\begin{align}
\label{s-t}
  s(w) &:= s \hbox{ if and only if } s\in Uw\hbox{ and } w^{-\ell}(s) \in \Lambda^c,
             \hbox{ for all }\ell\ge1,\nn\\
  t(w) &:= t \hbox{ if and only if } t\in Vw\hbox{ and }  w^{\ell}(t) \in \Lambda^c, \hbox{ for all }\ell\ge1.
\end{align}
For each  $w\in \tW_\Lambda$, denote by $w_\Lambda$ the piece of trajectory between $s(w)$  and $t(w)$,
\begin{align}
  \label{wa54}
  w_\Lambda&:= (z_0,z_1,\dots,z_{m(w)+1}), \hbox{ with }\, z_0=s(w),\, z_i=w^i(s),\, i=0,\dots,m(w)+1,
\end{align}
where $m(w)$ is the number of points of $w$ between the first and the last visits to $\Lambda$, in particular  $z_{m(w)+1}=t(w)\in \Lambda^c$. Note that $z_1$ and $z_{m(w)}$ belong to $\Lambda$ but intermediate points may not.
 By the definition of $Q^{\ri}$, the marginal weight of the piece of trajectory $w_\Lambda$ under $\mu^{\ri}_{\rho-\rho_c}$ is 
\begin{align}
  \label{wa53}
  f^{\ri}_{\rho-\rho_c}(w_\Lambda)
  &=(\rho-\rho_c) \,q(s(w))\,q(t(w))\,\prod_{\eta\in \tI w_\Lambda} p(\eta) \prod_{\eta' \in \tO w_\Lambda}p(\eta'),
  \end{align}
  where  $p$ was defined in \eqref{pes3}, and $q(x) := P^x(X_n\notin \Lambda, n\ge 1)$, where $P^x$ was defined in \eqref{psupx}. By definitions \eqref{dd90} and \eqref{dd91}, $\tI w_\Lambda$ are the red paths of $w_\Lambda$ joining yellow and/or green points in Fig.~\ref{gibbs-4} and $\tO w_\Lambda$ are the purple paths of $w_\Lambda$ between yellow/green points and the purple semi-infinite paths starting from green points. 
 
     \noindent\begin{minipage}{\linewidth}\begin{center}
    \includegraphics[scale=.5]{gibbs-33}\quad\quad %
    \includegraphics[scale=.5]{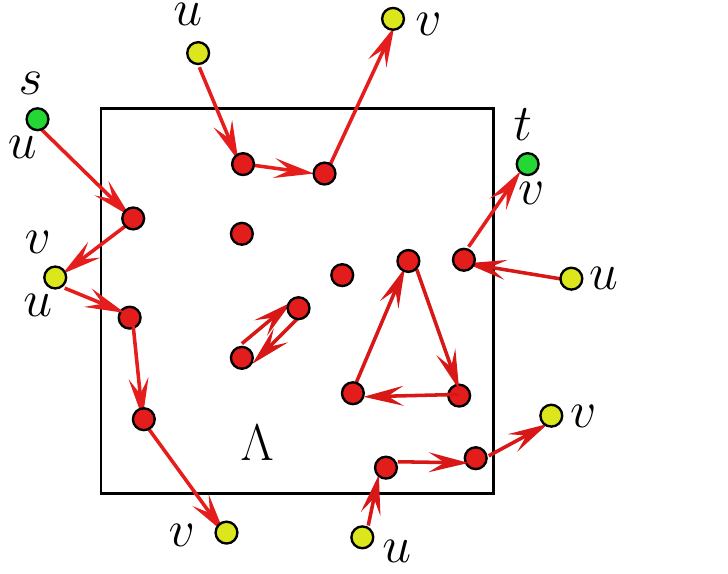}%
\\
\captionof{figure}{Decomposition of a spatial permutation $\Gamma$ intersecting $\Lambda$ in the supercritical case. Loops are decomposed as in Fig.~\ref{gibbs-3}. Only the finite trajectory between points labeled $s$ and $t$ of infinite paths is relevant for the intersection of the configuration with $\Lambda$, as shown in the picture on the right. \label{gibbs-4}}
\end{center}\end{minipage}

Given the positions and labels of the points $(U,V)$, the semi-infinite trajectories starting at $t(w)$ or finishing at $s(w)$ for $w\in \Gamma^{\ri}_{\rho-\rho_c}\cap \tW_\Lambda$ are independent of $w_\Lambda$, the piece of trajectory between these points, by the Markov property of the random walk. This and the properties of the Poisson processes imply that given  the positions and labels of the points $(U,V)$, those semi-infinite trajectories are independent of $\crF_{\tti \Lambda}$.

Given a spatial permutation $\Gamma$, denote by $\Gamma':= \Gamma \cap \BB_\Lambda$, the boundary loops and $\Gamma'':=\{w_\Lambda: w\in \Gamma\cap \tW_\Lambda\}$, the pieces of interlacements intersecting $\Lambda$. Using \eqref{fact2} and \eqref{wa53}, the weight $f _{\rho}$ of these configurations under $\mu_\rho$ are given by
  \begin{align}
    f_\rho (\Gamma'\cup\Gamma'')
    &=   f_\rho (\tI(\Gamma'\cup \Gamma''))\,  f_\rho (\tO(\Gamma'\cup\Gamma'')),\quad\text{where}\label{fact4}\\[3mm]
    f_\rho (\tI(\Gamma'\cup\Gamma''))&:=  \prod_{\eta\in\tI(\Gamma'\cup\Gamma'')} \,p(\eta),\nn\\
    f_{\rho} (\tO(\Gamma'\cup\Gamma''))&:=\prod_{\eta'\in\tO(\Gamma'\cup\Gamma'')}\,p(\eta')\,\,
      \prod_{w\in\Gamma''}  (\rho-\rho_c)\,q(s(w))\,q(t(w)).\label{fact5}
  \end{align}
An argument similar to \eqref{lsls1}-\eqref{S56} proves \eqref{markov53} for $\rho>\rho_c$. 
\end{proof}

\subsection{Gibbs measures}
\label{gm1}
We show here that $\mu_\rho$ is Gibbs for the specification associated to the canonical measure \eqref{finite-box}.

We first define the specification related to \eqref{finite-box}. 
Given a bounded Borel set $\Lambda$ we say that the spatial permutations $(\sX,\sigma)$  and $(\sY,\kappa)$ are $\Lambda$-compatible, and write $(\sX,\sigma) \stackrel{\Lambda}{\sim}(\sY,\kappa)$, if their outside components with respect to $\Lambda$ match, that is, if 
$  O (\sX,\sigma) = O (\sY,\kappa)$, see definition \eqref{IX}.
Given a spatial permutation $(\sY,\kappa)$, define 
\begin{align*}
  H_{\Lambda}\bigl((\sX,\sigma)\big| (\sY,\kappa)\bigr) &:= \sum_{x\in I\sX\,\cup \,U\sY} \|x-\sigma(x)\|^2,\qquad \text{for }(\sX,\sigma) \stackrel{\Lambda}{\sim}(\sY,\kappa).
\end{align*}
For $\lambda < 1$, and $\lambda=1$ if $d\ge 3$ let $G_{\Lambda,\lambda} \big(\cdot \big|(\sY,\kappa)\big)$ be the measure acting on bounded measurable test functions $g:\kG\to\R$ by 
\begin{align}
  G_{\Lambda,\lambda} \bigl(\,g \,\big|(\sY,\kappa)\bigr)
  &:=
    \sum_{m\ge n(U\sY)}
    \frac{1}{m!}\int_{\Lambda^m}
    \sum_{\substack{(\sX, \sigma)\stackrel{\Lambda}{\sim} (\sY,\kappa):\\I\sX=\{x_1,\dots, x_m\}}}
     g(\sX,\sigma)\,
   f_{\Lambda,\lambda} \bigl((\sX,\sigma)\big|(\sY,\kappa)\bigr)\,
    dx_1,\dots dx_m, \label{specif}
 \\[3mm]%
  f_{\Lambda,\lambda} \bigl((\sX,\sigma)\big|(\sY,\kappa)\bigr)
   &:= \frac1{Z_{\Lambda,\lambda}(\sY,\kappa)}
                                                               {\lambda^{n(I\sX)}\big((\alpha/\pi)^{d/2}\big)^{n(I\sX)+n(U\sY)}}
                   e^{- \alpha H_{\Lambda}((\sX,\sigma)| (\sY,\kappa))},\label{f70}
\end{align}
where we recall that $n(\sX)$ denotes the cardinality of $\sX$. The partition function is
\begin{align}
  Z_{\Lambda,\lambda}(\sY,\kappa)
  &:= \sum_{m\ge n(U\sY)}
      \frac{1}{m!}\,{\lambda^{m}\,\big((\alpha/\pi)^{d/2}\big)^{m+n(U\sY)}}
      \int_{\Lambda^m}
    \sum_{\substack{(\sX, \sigma)\stackrel{\Lambda}{\sim} (\sY,\kappa):\\I\sX=\{x_1,\dots, x_m\}}}
  e^{- \alpha H_{\Lambda}((\sX,\sigma)| (\sY,\kappa))}\,
    dx_1,\dots dx_m. \nn\label{zal}
\end{align}
For each $0<\lambda < 1$, and $\lambda=1$ if $d\ge 3$, the family of kernels  $(G_{\Lambda,\lambda})_\Lambda=\{G_{\Lambda,\lambda}:\Lambda $ bounded Borel set of $\R^d\}$ is called the specification associated to the Hamiltonian $H$ at fugacity $\lambda$. The next result says that the Gaussian random permutation $\mu_\rho$ is Gibbs with respect to this specification.
\begin{Proposition}[Gibbs measures]
  \label{t43}
  For $\lambda < 1$ and $\lambda=1$ if $d\ge 3$, the measure $\mu^{\ls}_\lambda$ is Gibbs for the specification $(G_{\Lambda,\lambda})_\Lambda$, that is,  it satisfies the DLR equations
  \begin{align}
    \mu^{\ls}_\lambda g
    = \int d\mu^{\ls}_\lambda(\sY,\kappa)\,
    G_{\Lambda,\lambda}\big(\,g\,\big|(\sY,\kappa)\big).
  \end{align}
Furthermore, if $d\ge 3$, for all $\rho\ge \rho_c$  the measure $\mu_\rho=\mu^{\ls}_1 *  \mu^{\ri}_{\rho-\rho_c}$ is Gibbs for the specification $(G_{\Lambda,1})_\Lambda$.
\end{Proposition}
\begin{proof}
When there are no external boundary points, that is when $U(\sY,\kappa)=V(\sY,\kappa)=\emptyset$, the specification \eqref{specif} equals the grand-canonical measure \eqref{23c}, which by Proposition \ref{p24} is the Gaussian loop soup at fugacity $\lambda$ intersected with $\DD_\Lambda$. When  $U(\sY,\kappa) \neq\emptyset$, the specification \eqref{specif} factorizes as the product of the grand-canonical measure \eqref{23c} and the distribution of crossings of $\Lambda$. We conclude that the specification coincides with the conditional probabilities:
\begin{align}
  G_{\Lambda,\lambda}(g|(\sY,\kappa)) &= \mu_\rho(g|\crF_{\tto \Lambda})(\sY,\kappa) =\mu_\rho(g|\crF_{\ttb \Lambda}) (U,V)(\sY,\kappa) , \notag
\end{align}
where $\lambda=\lambda(\rho)$ for $d \le 2$, or $d\ge3$ and $\rho\le\rho_c$, and $\lambda=1$ for $d\ge3$ and  $\rho>\rho_c$.
\end{proof}

\begin{Remark}\rm
\label{non-unique}
All supercritical measures $\mu_\rho$, for $\rho\ge \rho_c$ are Gibbs for the same specification $(G_{\Lambda,1})_\Lambda$. The different point densities inside $\Lambda$ arise from the boundary conditions. Indeed, the number of paths crossing $\Lambda$ is equal to the cardinality of $U$ (or $V$). Part of those $U$-$V$ points belong to loops of the critical loop soup, and have the same law for all supercritical densities. The contribution to the extra density is given by the $U$-$V$ points belonging to infinite trajectories, whose number increases with $\rho$. The mean number of points in $\Lambda$ from trajectories that connect $U$-$V$ points arising from the random interlacements is  $(\rho-\rho_c)|\Lambda|$.

Consider now a large box $\Lambda_N$ with volume $N$. The trace left in $\Lambda_N$ by an interlacement trajectory that intersects this box typically comprises $O(N^{2/d})$ points. On the other hand, for supercritical densities $\rho>\rho_c$, a macroscopic loop in the canonical measure with $\rho N$ points in $\Lambda_N$  consists of $O(N)$ points. The heuristic explanation for this discrepancy is that, in the conjectured thermodynamic limit, macroscopic cycles will become infinite cycles going through $\infty$, and they will break up into several (infinitely many)
random interlacement trajectories. 
\end{Remark}

\section{The point marginal and the boson point process}
\label{s4}
We show that the point marginal of the Gaussian random permutation coincides with the boson point process associated to the free Bose gas. Let $\lambda(\rho)$ be defined as in \eqref{sol-lambda}, and denote the point marginal of the Gaussian random permutation \eqref{grp7} by 
\begin{align}
  \label{grpoint}
  \nu_\rho :=
  \begin{cases}
    \nu^{\ls}_{\lambda(\rho)},
    &\hbox{if } d\le2 \text{ and } \rho>0,\;\text{ or }\; d\ge3\;\text{and } \rho\le \rho_c, \\[2mm]
    \nu^{\ls}_1 * \nu^{\ri}_{\rho-\rho_c},
    &\hbox{if }\;d\ge3\;\text{and } \rho> \rho_c,
  \end{cases}
\end{align}
where $\nu^{\ls}_\lambda$ and $\nu^{\ri}_\rho$ are defined in \eqref{loop points} and \eqref{chi}, respectively.

\paragraph{The boson point process}
\emph{Subcritical and critical cases}. In \cite{TI} Tamura and Ito show that the point process for the free 
Bose gas in $\R^d$ at density $\rho>0$ if $d=1, 2$, and $\rho\le \rho_c$ if $d\ge 3$, and fugacity $\lambda=\lambda(\rho)$, is the point process described 
by Theorem 1.2 in Shirai and Takahashi \cite{ST}. Its law, denoted $\nu^{\st}_\lambda$, 
has Laplace functional given by
\begin{align}
  \label{st77}
  \cL^{\st}_\lambda(\phi) = \Det( I + \cK^\lambda_\phi )^{-1},
\end{align}
where $\phi:\R^d\to \R$ is a non-negative function with compact support, $\cK^\lambda_\phi$ is the linear operator on $L^2(\R^d)$ with kernel
\begin{align*}
\sqrt{1-e^{-\phi(x)}}\, K_\lambda(x,y) \sqrt{1-e^{-\phi(y)}}, \qquad K_\lambda(x,y) \text{ as in \eqref{Klambda}}, 
\end{align*}
and $\Det$ is the Fredholm determinant. They also show that the $n$-point correlation functions of $\nu^{\st}_\lambda$ are given by
\begin{align}
  \label{ST-lscorr}
  \varphi^{\st}_\lambda(x_1,...,x_n) = \perm \left(K_\lambda(x_i,x_j)\right)_{i,j=1}^n.
\end{align}
\emph{Supercritical case}.  
Bose-Einstein condensation occurs for the free Bose gas when $d \geq 3$ and $\rho>\rho_c$.
In this case Tamura and Ito \cite{TI2} describe the boson point process $\nu_\rho$ as the superposition of the process
$\nu^{\st}_{1}$ defined above with a process $\nu^{\TI}_{\rho-\rho_c}$,  with Laplace functional
\begin{align}
  \label{super11}
  \cL^{\TI}_{\rho-\rho_c}(\phi) = \exp \left( -(\rho-\rho_c) \big\langle 
\sqrt{1-e^{-\phi}}, (I+\cc K^1_\phi)^{-1} \sqrt{1-e^{-\phi}} \big\rangle \right),
\end{align}
where
$\phi:\R^d\to \R$ is a non-negative function with compact support, $\cc K^1_\phi$ is the operator defined in \eqref{st77}, and $\langle \cdot, \cdot\rangle$ is the usual
inner product in $L^2(\R^d)$.

Summarizing, Tamura and Ito define the \emph{boson point process}  by
\begin{align}
  \label{bosonp}
  \nu^{\boson}_\rho :=
  \begin{cases}
    \nu^{\st}_{\lambda(\rho)},
    &\hbox{if } d\le2 \text{ and } \rho>0,\;\text{ or }\; d\ge3\;\text{and } \rho\le \rho_c, \\[2mm]
    \nu^{\st}_1 * \nu^{\TI}_{\rho-\rho_c},
    &\hbox{if }\;d\ge3\;\text{and } \rho> \rho_c.
  \end{cases}
\end{align}

\begin{Proposition}[Equivalence of point  processes] \label{teo1}
  For any $\rho>0$ we have
  \begin{align*}
    \nu_\rho = \nu^{\boson}_\rho,
  \end{align*}
where the processes are defined in \eqref{grpoint} and \eqref{bosonp}, respectively.
\end{Proposition}

\begin{proof} \emph{Subcritical and critical cases}. 
By  Proposition \ref{ls-correl}, $\nu^{\ls}_\lambda$ and $\nu^{\st}_\lambda$ have the same $n$-point correlation functions, $\varphi^{\ls}_\lambda=\varphi^{\st}_\lambda$.
By a theorem of Lenard \cite{MR323270}, the measures coincide if the correlation functions satisfy the following bounds
  \begin{align}
    \label{len1}
    \sup_{x_1,\dots,x_n\in\R^d} |\varphi^{\ls}_\lambda(x_1,...,x_n)| < (c \,n^2)^n, \quad  n=1,2,\dots.
  \end{align}
  Since $K_\lambda(x,y)\le K_\lambda(0,0)<\infty$ under the conditions of the proposition, we have
  \begin{align}
    \label{len2}
    \varphi^{\ls}_\lambda(x_1,...,x_n) = \perm \left(K_\lambda(x_i,x_j)\right)_{i,j=1}^n \le n! \bigl(K_\lambda(0,0)\bigr)^n.
  \end{align}
This satisfies \eqref{len1}.
Alternatively, one can simply observe that the Laplace functional of $\nu^{\ls}_\lambda$ provided by \eqref{F1}-\eqref{F2} 
coincides with the Laplace functional of $\nu_\lambda^{\st}$, computed in  Theorem 4.1 of \cite{ST}.

\emph{Supercritical case}. 
 Let $\phi:\R^d\to \R$ be a non-negative function with compact support with
  $K_1(0,0) \|e^{-\phi}-1\|_1<1$. This condition allows to expand 
  $(I+\cc K^1_\phi)^{-1} = \sum_{n=0}^\infty (-1)^n (\cc K^1_\phi)^n$ and replace in 
  \eqref{super11} to obtain the Laplace functional $\cL^{\ri}_\rho(\phi)$ computed in \eqref{laplace1}.  
\end{proof}

\begin{Remark}\rm
\label{iso2}
For $d\ge 3$ and $\rho>\rho_c$, Eisenbaum \cite{Eis} proves that  the supercritical boson point process $\nu^{\boson}_\rho=\nu^{\st}_1 * \nu^{\TI}_{\rho-\rho_c}$ is a Cox process with (random) intensity 
\begin{align*}
\tfrac12\bigl(\Phi_1+\sqrt{2(\rho-\rho_c)}\bigr)^2+\tfrac12 \Psi_1^2,
\end{align*}
where $\Phi_1$ and $\Psi_1$ are independent, centered Gaussian fields with covariance $\big(K_1(x,y),\, x,y \in \R^d\big)$.
\end{Remark}

\begin{Remark}\rm
In view of Remarks \ref{iso1} and \ref{iso2}, by Campbell's formula, given a bounded, compactly supported test function $f:\R^d\to \R$, we get
\begin{align}
&\int\nu_\rho(d\sX)\sum_{x\in \sX}f(x) =E\Big[\int_{\R^d} f(x) \big( \tfrac12 \Phi^2_\lambda(x)+\tfrac12 \Psi^2_\lambda(x) \big) dx \Big] \quad\ \ \text{if }\quad\ \  
\begin{cases} d\le 2,\, & \rho>0\\ d\ge 3,\, & 0\le \rho\le \rho_c, \end{cases} \label{iso1.1}\\[2mm]
&\int{\nu_\rho(d\sX)}\sum_{x\in \sX}f(x)=E\Big[\int_{\R^d} f(x)\Big( \tfrac12\big(\Phi_1(x)+\sqrt{2(\rho-\rho_c)}\big)^2+\tfrac12 \Psi_1^2(x)\Big)\,dx\Big]\quad \text{if}\quad d\ge 3, \,\rho>\rho_c. \label{iso2.1}
\end{align}
Identities \eqref{iso1.1} and \eqref{iso2.1} may be interpreted as versions, for the boson point process, of the occupation times formulas derived by Le Jan \cite{MR2675000} for the loop soup, and Sznitman \cite{MR2892408} for the random interlacements.
\end{Remark}

\section{Conclusions}
\label{Sconclusions}

We have shown that the Gaussian random permutation, consisting of the superposition of independent realizations of a Gaussian loop soup and a Gaussian random interlacements process, is Gibbs for the specification induced by the Feynman's representation of the free Bose gas. The random interlacements are present only for point densities $\rho>\rho_c$. We then used the fact that the two processes are independent Poisson processes to show that the point marginal of the Gaussian random permutation coincides with the boson point process, studied previously in the literature.  

The interacting Bose gas proposed by Feynman \cite{Feynman} can be seen as a measure on the set of continuous trajectories $\uB:=(B_1,\dots,B_N)$, $B_i=(B_i(s))_{s\in[0,\beta]}$, where $\beta=\frac1{2\alpha}$. The canonical measure $\EuScript M_\Lambda$ associated to a pair interaction potential $\cV(x,y)$ acts on a test function $g$ by
\begin{align}
  \label{interact}
\EuScript M_\Lambda g &:= \frac1{Z_{\Lambda,N}}\frac{1}{N!}\sum_{\sigma \in \cS_N} \int_{\Lambda^n} d\ux\; \int\, P^{\ux,\sigma}(d\uB)\,g(\uB)\, \exp\bigl[-\tfrac1{2\beta} H(\ux,\sigma)-\tsum_{i<j}\,\tint_0^\beta\, \cV\bigl(B_i(s),B_j(s)\bigr)\,ds\bigr] ,  \notag
\end{align}
where $P^{\ux,\sigma}$ denotes the law of $N$ independent Brownian bridges such that $B_i(0)=x_i$  and $B_i(\beta)=x_{\sigma(i)}$, for each $i$, and the partition function  $Z_{\Lambda,N}$ is given in \eqref{Zlambda}, see \cite{zbMATH05504000}. An example of interaction is the hard-core potential $\cV(x,y) = \infty$ if $\|x-y\|\le b$, and zero otherwise.

The challenge is to construct a Gibbs measure in $\R^d$ for the specification associated to $\EuM_\Lambda$. For the ideal Bose gas, $\cV\equiv0$, the Gibbs measure can be easily constructed from the spatial Gaussian permutation $(\sX,\sigma)$ with law \eqref{grp7}, by considering a collection of independent Brownian bridges $\bigl(B^x)_{x\in\sX}$, where $(B^x(s))_{s\in[0,\beta]}$ is a Brownian bridge between $x$ and $\sigma(x)$ at time $\beta$. The resulting measure $\EuM^{\free}$
is Gibbs for the specification associated to $\EuM_\Lambda$ with $\cV\equiv0$. Concatenating the bridges,  one obtains a Brownian loop soup where the duration of each loop is a multiple of $\beta$, superposed in the supercritical case with a Brownian interlacement process. One may conjecture that there are pair potentials $\cV$ (say, bounded below and vanishing when $\|x-y\|$ exceeds some value $b$) such that there is a corresponding Gibbs measure for $\EuM_\Lambda$ whose realizations are local perturbations of the realizations of  $\EuM^{\free}$.

\section{Acknowledgments} We thank the referees for their comments and suggestions.

\bibliographystyle{amsplain}

\def\cprime{$'$}
\providecommand{\bysame}{\leavevmode\hbox to3em{\hrulefill}\thinspace}
\providecommand{\MR}{\relax\ifhmode\unskip\space\fi MR }
\providecommand{\MRhref}[2]{%
  \href{http://www.ams.org/mathscinet-getitem?mr=#1}{#2}
}
\providecommand{\href}[2]{#2}

\end{document}